\documentclass[a4paper,12pt]{article}

\usepackage{setspace}
\usepackage{times} 
\usepackage{color}

\usepackage[authoryear]{natbib}

\usepackage[top=1.25in, bottom=1.25in, left=1in, right=1in]{geometry}

\usepackage{rotating}

\usepackage{amsmath,amsthm,amsfonts,amssymb,bm} 

\newcommand\fullwidthdisplay{\displayindent0pt \displaywidth\columnwidth}
\AtBeginDocument{
\everydisplay\expandafter{\expandafter\fullwidthdisplay\the\everydisplay}
}

\usepackage{float}

\usepackage{makeidx} 
\usepackage{listings} 
\usepackage[colorlinks,
linkcolor=black,
anchorcolor=black,
citecolor=black
]{hyperref}

\usepackage{graphicx,psfrag} 
\usepackage{epstopdf}
\usepackage{multirow}
\usepackage{booktabs}

\theoremstyle{plain}{

\newtheorem{proposition}{Proposition}[section]

\newtheorem{corollary}{Corollary}[section]
\newtheorem{lemma}{Lemma}[section]

\newtheorem{Assumption}{Assumption}[section]}

\theoremstyle{definition}{

}

\def\E{{\mathbb{E}}}
\def\P{{\mathbb{P}}}

\numberwithin{equation}{section}

\usepackage{hyperref}
\hypersetup{
colorlinks=true, 
breaklinks=true, 
urlcolor=blue, 
linkcolor=red, 
citecolor=blue, 
pdftitle={}, 
pdfauthor={}, 
pdfsubject={} 
}

\begin{document}

\author{Pengyu Wei\thanks{\rm Division of Banking and Finance, Nanyang Business School, Nanyang Technological University, Singapore. E-mail: pengyu.wei@ntu.edu.sg.}
\and
Zuo Quan Xu\thanks{\rm Department of Applied Mathematics, The Hong Kong Polytechnic University, Hong Kong. Email: maxu@polyu.edu.hk.}
} 
\title{Dynamic growth-optimum portfolio choice\\ under risk control} 

\date{\today}
\maketitle

\begin{abstract}
This paper studies a mean-risk portfolio choice problem for log-returns in a continuous-time, complete market. 
This is a growth-optimal problem with risk control. 
The risk of log-returns is measured by weighted Value-at-Risk (WVaR), which is a generalization of Value-at-Risk (VaR) and Expected Shortfall (ES).
We characterize the optimal terminal wealth up to the concave envelope of a certain function, and obtain analytical expressions for the optimal wealth and portfolio policy when the risk is measured by VaR or ES. In addition, we find that the efficient frontier is a concave curve that connects the minimum-risk portfolio with the growth optimal portfolio, as opposed to the vertical line when WVaR is used on terminal wealth. Our results advocate the use of mean-WVaR criterion for log-returns instead of terminal wealth in dynamic portfolio choice.

\bigskip

\end{abstract}

\clearpage

\section{Introduction}
The growth-optimal portfolio (GOP) is a portfolio which has a maximal expected growth rate (namely log-return) over any time horizon. As the GOP can be usually tracked to the work \cite{K1956}, it is also called the ``Kelly criterion''. The GOP can also be obtained by maximizing log-utility which has a longer history. 
As the name implies, it can be used to maximize the expected growth rate of a portfolio. Indeed, it performs in some sense better than any other significantly different strategy as the time horizon increases. 
 
Over the past half century, a lot of papers have investigated the GOP. 
In theory and practice, the GOP has widely applications in a large number of areas including portfolio theory, utility theory, game theory, information theory, asset pricing theory, insurance theory. For instance, to name a few of them in the recent studies in the literature, \cite{A2000} study asset pricing problems in incomplete market; \cite{R2004} considers optimal investment problems; \cite{T2000} applies it to casino games.

We want to emphasize that the GOP ignores the relevant risk when maximizing the expected growth rate of a portfolio. 
It is the seminal work \cite{markowitz1952portfolio} that takes the trade-off between the portfolio return and its risk into consideration when an investor chooses a portfolio. \cite{markowitz1952portfolio} suggests to use variance to measure the risk. Since then, the mean-variance theory becomes one of the most dominant financial theories in the realm of portfolio choice. Besides variance, alternative risk measures have been proposed to measure the risk for portfolio choice. Research along this line includes \cite{rockafellar2000optimization}, \cite{campbell2001optimal}, \cite{rockafellar2002conditional}, \cite{alexander2002economic}, \cite{alexander2004comparison},
\cite{jin2006note}, and \cite{adam2008spectral},
where the authors study single-period mean-risk portfolio selection with various risk measures, such as semi-variance, value-at-risk (VaR), expected shortfall (ES), and spectral risk measures.

There also have numerous extensions of the mean-risk portfolio optimization from the single-period setting to the dynamic, continuous-time one 
\citep[e.g.][]{zhou2000continuous,bielecki2005continuous,jin2005continuous,basak2010dynamic,he2015dynamic,zhou2017dynamic,gao2017dynamic,dai2021dynamic,he2021mean}. In particular, \cite{he2015dynamic} study a continuous-time mean-risk portfolio choice when risk is measured by the weighted Value-at-Risk (WVaR) but their results are rather pessimistic. The WVaR is a quantile-based risk measure that generalizes VaR and ES, two popular risk measures in quantitative risk management. 
They find that, when using WVaR (including VaR and ES) on terminal wealth to measure portfolio risk, the mean-risk model is prone to be ill-posed (i.e., the optimal value is infinite) and the investor tends to take infinite leverage on risky assets, leading to extremely risk-taking behaviors. Furthermore, the optimal risk is independent of the expected terminal target so the efficient frontier is a vertical line on the mean-WVaR plane. Their results suggest that the mean-WVaR model is an improper modeling of the trade-off between return and risk, when the WVaR is applied to the terminal wealth.

This paper proposes a continuous-time portfolio choice model with mean-WVaR criterion for portfolio log-returns, as opposed to the mean-WVaR criterion for terminal wealth in \cite{he2015dynamic}. The motivation is two-fold. First, the mean-risk criterion for log-returns is consistent with Markowitz's original idea to use mean and variance on portfolio returns. We consider a growth-optimal problem with risk control. 
Moreover, many single-period mean-risk models use risk measures on portfolio returns in the literature \citep[e.g.][]{alexander2002economic,alexander2004comparison,alexander2006does,jin2006note,adam2008spectral}. 
However, there is a discrepancy between single-period and dynamic portfolio choice models, as the latter typically considers mean-risk criterion for terminal wealth; an exception is \cite{dai2021dynamic} who study continuous-time mean-variance portfolio choice for portfolio log-returns. We similarly adopt the mean-WVaR criterion for log-returns which naturally generalize the single-period return when returns are continuously compounded. Second, such a criterion conquers the ill-posedness of the model in \cite{he2015dynamic}. As noted in \cite{he2015dynamic}, the mean-WVaR criterion for terminal wealth is essentially a linear program for the quantile function of terminal wealth. This linearity, in turn, leads to the optimal terminal wealth's quantile function being ``corner points'', leading to extreme risk-taking behaviors. By contrast, our mean-WVaR criterion for log-returns is not linear in the quantile function of terminal wealth, and thus conquers the ill-posedness.

In a continuous-time, complete market framework, we solve the mean-WVaR portfolio choice for log-returns with the help of the so-called quantile formulation, developed in a series of papers \citep[e.g.][]{schied2004neyman,carlier2006law,jin2008behavioral,he2011portfolio,carlier2011optimal,xia2016arrow,xu2016note}. When risk is measured by a general WVaR risk measure, we characterize the optimal terminal wealth up to the concave envelope of a certain function through a detailed and involved analysis. When risk is measured by VaR or ES, two special cases of WVaR, we derive analytical expressions for the optimal terminal wealth and portfolio policy. The optimal terminal wealth turns out to be closely related to the growth optimal portfolio: the investor classifies market scenarios into different states, in which the terminal payoff can be constant, a multiple or fraction of the growth optimal portfolio. Furthermore, we obtain the efficient frontier, which is a concave curve that connects the minimum-risk (min-risk) portfolio with the growth optimal portfolio, as opposed to the vertical line in \cite{he2015dynamic}. Our model allows for a meaningful characterization of the risk-return trade-off and may serve as a guideline for investors to set a reasonable investment target. Although \cite{he2015dynamic} provides a critique of using WVaR to measure risk, our results advocate that it is more appropriate to use WVaR, in particular, VaR and ES, on portfolio log-returns instead of terminal wealth for dynamic portfolio choice.

The rest of the paper is organized as follows. In Section \ref{sec:model}, we propose a mean-WVaR portfolio choice problem for portfolio log-returns. We solve the problem in 
Section \ref{sec:solution} by quantile optimization method. Section \ref{sec:examples} presents optimal solutions and efficient frontiers when the risk is measured by VaR or ES. Some new financial insights and a comparation to the existing work are presented as well. Some concluding remarks are given in Section \ref{sec:conclusion}. Appendix \ref{sec:A1} contains three useful lemmas. All remaining proofs are placed in Appendix \ref{sec:A2}.

\section{Mean-WVaR portfolio choice model}\label{sec:model}

\subsection{Financial market}
Let $T>0$ be a given terminal time and $(\Omega,\mathcal{F}, \{ \mathcal{F}_t \}_{0 \le t \le T} ,\mathbb{P})$ be a filtered probability space, on which is defined a standard one-dimensional Brownian motion 
$\{ W_t \}_{0\le t\le T}$. It is assumed that $\mathcal{F}_t=\sigma \{ W_s, 0\le s\le t \}$ augmented by all $\mathbb{P}$-null sets and that 
$\mathcal{F}=\mathcal{F}_T$ is $\mathbb{P}$ complete. 

We consider a Black-Scholes market in which there are a risk-free asset and a risky asset (called stock). The risk-free asset pays a constant interest rate $r>0$ and the stock price $S$ follows a geometric Brownian motion
$$\frac{dS_t}{S_t}=\mu dt+\sigma dW_t, $$
where $\mu$ and $\sigma$, the appreciation rate and volatility of the stock, are positive constants. There exists a unique positive state price density (pricing kernel) process $\xi$\footnote{With additional assumptions on $\xi_T$, our main results can be extended to a general complete market with stochastic investment opportunities.} satisfying
\begin{equation}\label{eq:xi}
\frac{d\xi_t}{\xi_t}=-r dt-\theta dW_t, \quad \xi_{0}=1,
\end{equation}
where $\theta=(\mu-r)/\sigma$ is the market price of risk in the economy. Therefore the market is complete. 

Consider an economic agent with an initial endowment $x>0$ and faced an investment time horizon $[0,T]$. The agent chooses a dynamic investment strategy $\pi_t$, which represents the dollar amount invested in the stock at time $t$. Assume the trading is continuous in a self-financing fashion and there are no transaction costs. The agent's wealth process $X_t$ then follows a stochastic differential equation 
\begin{equation}\label{eq:budget}
dX_t=\left[ rX_t+(\mu-r) \pi_t \right] dt+\sigma \pi_t dW_t, ~ X_0=x.
\end{equation}
The portfolio process $\pi_t$ is called an admissible portfolio if it is $\{ \mathcal{F}_t \}_{0 \le t \le T}$ progressively measurable with $\int_0^T \pi_t ^2 dt < \infty, a.s.$, and the corresponding terminal wealth satisfies $X_T \ge 0, a.s$.

Let $R_T$ be the continuously compounded return (log-return) over the horizon $[0,T]$, i.e.,
\begin{equation}\label{eq:log-return}
R_T=\frac{1}{T} \ln \frac{X_T}{x}.
\end{equation}
By convention, we define
$$\ln 0=\lim_{s \downarrow 0} \ln s=-\infty \mbox{ and } e^{-\infty}=\lim_{s \downarrow-\infty}e^s=0.$$

\subsection{Risk measure}

We now introduce a risk measure that will be used in the portfolio choice model. In this paper, we focus on the weighted VaR (WVaR) risk measure proposed by \cite{he2015dynamic}, which is a generalization of Value-at-Risk (VaR) and Expected Shortfall (ES), and encompasses many well-known risk measures that are widely used in finance and actuarial sciences, such as spectral risk measures and distortion risk measures; see \cite{wei2018risk} for a review. 

For any $\mathcal{F}_T$-measurable random variable $X$, let $F_X$ denote its cumulative distribution function; and let $G_X$ denote its quantile function defined by
\begin{equation*}
	G_X(z)=\inf \{x \in \mathbb{R} : F_X(x) > z \}=\sup \{x\in \mathbb{R} : F_X(x) \le z \}, ~z \in [0, 1),
\end{equation*}
with the convention $G_X(1)=\lim_{z \uparrow 1} G_X(z)$. The quantile function $G_X$ is non-decreasing, right-continuous with left limits (RCLL). 

The WVaR risk measure for $X$ is defined as
\begin{equation}\label{eq:wvar}
\rho _{\Phi} (X)=-\int _{[0,1]} G_X(z) \Phi (dz),
\end{equation}
where $\Phi\in P[0,1]$ and $P[0,1]$ is the set of all probability measures on $[0,1]$. 

The WVaR is a law-invariant comonotonic additive risk measure, and it covers many law-invariant coherent risk measures; see \cite{he2015dynamic} for a more detailed discussion. If $\Phi$ is the Dirac measure at $\alpha$, i.e., $\Phi (A)=\mathbf{1}_{\alpha \in A}$, for all $A \subset [0,1]$, then the corresponding WVaR measure becomes the VaR at $\alpha$, in other words,
\begin{equation*}
\rho _{\Phi} (X)=\text{VaR}_{\alpha}(X)=-G_{X}(\alpha).
\end{equation*}
If $\Phi$ admits a density $\phi (z)=\frac{1}{\alpha}\mathbf{1}_{z \le \alpha}, ~ \forall z \in [0,1]$, then the corresponding WVaR measure becomes the ES, i.e.,
\begin{equation*}
\rho _{\Phi} (X)=\text{ES}_{\alpha} (X)=-\frac{1}{\alpha} \int_0^{\alpha} G_{X} (z)dz.
\end{equation*}

In the original paper of \cite{he2015dynamic}, WVaR is applied to measure the risk of a portfolio's terminal wealth. In this paper, we propose to apply WVaR to the portfolio's log-return instead of its terminal wealth. Let $X_T$ be the terminal wealth of a portfolio and $R_T$ be the log-return of $X_T$. Due to the monotonicity of logarithm functions, the quantile function of $R_T$ is
\begin{equation*}
	G_{R_T}(z)=\frac{1}{T} \ln \frac{ G_{X_T}(z)}{x}, ~ z \in [0,1].
\end{equation*}
Therefore, the WVaR of $R_T$ can be expressed as 
\begin{equation}\label{eq:wvar log-return}
	\rho _{\Phi} (R_T)=-\int _{[0,1]} \frac{1}{T} \ln \frac{ G_{X_T}(z)}{x} \Phi (dz)=-\frac{1}{T} \int _{[0,1]} \ln G_{X_T}(z) \Phi (dz)+\frac{1}{T} \ln x.
\end{equation}

However, the extension from terminal wealth to log-return is not straightforward as the integral in \eqref{eq:wvar log-return} may not be well-defined since $X_{T}$ may take the value of 0 with positive probability. Let
\begin{equation*}
	\begin{aligned}
		\mathbb{G}=\Big\{G(\cdot) \colon [0, 1] \to [0,+\infty], ~ &G\text{ is nondecreasing and RCLL on 	[0,1], }\\
		& \text{ left-continuous at } 1, \text{ and finite-valued on } [0,1)\Big\}
	\end{aligned}
\end{equation*}
be the set of quantile functions of all non-negative random variables, which include all terminal wealth of admissible portfolios. For any $G \in \mathbb{G}$ and $\Phi\in P[0,1]$, the integral $\int _{[0,1]} \ln G(z) \Phi (dz)$ is not well-defined if $G(s)=0$ for some $s \in[0,1]$ such that $\Phi ( [0,s] )>0$. Define
\begin{equation*}
		\mathbb{G}_{\Phi}=\big\{G \in \mathbb{G} \colon G(s)>0 \text{ if } \Phi ( [0, s] )>0, ~ \forall s \in[0,1] \big\}.
\end{equation*}
We set 
\begin{equation}\label{eq:-infty integral}
\int _{[0,1]} \ln G(z) \Phi (dz)=-\infty, ~ \forall G \in \mathbb{G} \setminus \mathbb{G}_{\Phi}.
\end{equation}

Intuitively, if the terminal wealth of a portfolio is $0$ (so that the log-return is $-\infty$) in some states, and the weighting measure $\Phi$ assigns non-zero weighs to these states, then the WVaR of the log-return is assumed to be $-\infty$. In particular, $\E [ R_T]=-\infty$ if $\mathbb{P} (X_T=0)>0$.\footnote{It is straightforward to verify $\E \left[ \max \left(R_T,0 \right) \right]<\infty$, given that $\xi_T$ is log-normally distributed. }

\subsection{Portfolio choice model}
We assume the agent chooses a dynamic portfolio strategy in the period $[0,T]$ to maximize the expected log-return while minimizing the risk of the portfolio's log-return. The risk is evaluated by a WVaR risk measure $\rho _{\Phi}$ on the portfolio's log-return $R_T$. Specifically, we consider the following dynamic portfolio choice problem
\begin{equation}\label{prob:original}
\begin{aligned}
\max _{\pi_t} ~ &~ \lambda \E [ R_T]-\rho _{\Phi} (R_T)\\
\text{subject to} ~ &~dX_t=\left[ rX_t+(\mu-r) \pi_t \right]dt+\sigma \pi_t dW_t, ~X_{T}\geq 0, ~ X_0=x,\\
&R_T=\frac{1}{T} \ln \frac{X_T}{x},
\end{aligned}
\end{equation}
where $\lambda \ge 0$ is a ``risk-tolerance" parameter that reflects the investor's tradeoff between return and risk.
This is a stochastic control problem, but not standard (namely, unlike those in \cite{yongzhou1999}) due to the existence of the nonlinear probability measure $\Phi$.

In view of the standard martingale method, e.g., \cite{karatzas1998methods}, we can first solve the following static optimization problem\footnote{This formulation implies that the optimal log-return $R_{T}$ is independent of $x$. }
\begin{equation}\label{prob:martingale}
\begin{aligned}
\max _{X_T \in \mathcal{F}_T} ~ &~ \lambda \E [ R_T]-\rho _{\Phi} (R_T) \\
\text{subject to} ~ &~\E [\xi_T X_T] \le x, ~X_{T}\geq 0, \\
&~R_T=\frac{1}{T} \ln \frac{X_T}{x},
\end{aligned}
\end{equation}
where $\xi_{T}$ is given by \eqref{eq:xi}. Then apply backward stochastic control theory to derive the corresponding optimal portfolio strategy $\pi_{t}$.

The optimization problem \eqref{prob:martingale} nests two special cases.

\begin{description}
\item[Case $\lambda=0$.] In this case the investor minimizes the risk without any consideration of the expected log-return, and solves the following minimum-risk problem
\begin{equation}\label{prob:min risk}
\begin{aligned}
\min _{X_T \in \mathcal{F}_T} ~ &~ \rho _{\Phi} (R_T)\\
\text{subject to} ~ &~\E [\xi_T X_T] \le x, ~X_{T}\geq 0, \\
&~R_T=\frac{1}{T} \ln \frac{X_T}{x}.
\end{aligned}
\end{equation}
The resulting portfolio is termed the min-risk portfolio.

\item[Case $\lambda=\infty$.] In this case the investor maximizes the expected log-return without any consideration of the risk. This is the so-called growth-optimal problem
\begin{equation}\label{prob:growth}
\begin{aligned}
\max _{X_T \in \mathcal{F}_T} ~ & \E [ R_T] \\
\text{subject to} ~ &\E [\xi_T X_T] \le x, ~X_{T}\geq 0, \\
&R_T=\frac{1}{T} \ln \frac{X_T}{x}.
\end{aligned}
\end{equation}

The optimal solution to \eqref{prob:growth} is well-known in the literature, i.e., the growth-optimal portfolio (or \cite{K1956} strategy) given by
\begin{equation}\label{eq:growth}
X_{\textrm{Kelly}}=\frac{x}{\xi_T}.
\end{equation}
The corresponding log-return is 
\begin{equation*}
R_{\textrm{Kelly}}=\frac{1}{T} \ln \frac{X_{\textrm{Kelly}}}{x}=-\frac{1}{T} \ln \xi_T,
\end{equation*}
and its expected value is
\begin{equation*}
\E [R_{\textrm{Kelly}}]=-\frac{1}{T} \E [\ln \xi_T]=r+\frac{\theta^2}{2}.
\end{equation*}
\end{description}

\section{Quantile formulation and optimal solution}\label{sec:solution}

In this section, we solve the optimization problem \eqref{prob:martingale} for $0\le \lambda<\infty$. 

If $\Phi (\{ 1 \})>0$, then $\rho _{\Phi} (R_{T})=-\infty$. If $\Phi$ is the the uniform measure on $[0,1]$, then $\rho _{\Phi} (R_T)=-\E [R_T]$ and the growth optimal portfolio \eqref{eq:growth} is optimal to \eqref{prob:martingale}. 
To exclude these trivial cases, we make the following assumption on $\Phi$ from now on.
\begin{Assumption}\label{assumption:phi}
	$\Phi (\{ 1 \})=0$ and $\Phi$ is not the uniform measure on $[0,1]$.
\end{Assumption}

The objective in \eqref{prob:martingale} is based on the quantile function of the log-return; thus, the standard convex duality method is not readily applicable. To overcome this difficulty, we employ the quantile formulation, developed in a series of papers
including \cite{schied2004neyman}, \cite{carlier2006law}, \cite{jin2008behavioral}, \cite{he2011portfolio}, \cite{carlier2011optimal}, \cite{xia2016arrow}, and \cite{xu2016note}, to change the decision variable in \eqref{prob:martingale} from the terminal wealth $X_T$ to its quantile function. This allows us to recover the hidden convexity of the problem and solve it completely. 

We first show that the budget constraint in \eqref{prob:martingale} must hold with equality and the objective function is improved with a higher level of initial wealth.

\begin{lemma}\label{lemma:3.1}
	If $X_T^{*}$ is an optimal solution to the problem \eqref{prob:martingale}, then $ \E [\xi_T X_T^{*}]=x$.
\end{lemma}
\noindent
All the proofs of our results are put in Appendix \ref{sec:A2}.

Denote by $F_\xi$ and $G_\xi$ the distribution and quantile functions of $\xi_T$, respectively. With slight abuse of notation, we suppress the subscript $T$ when there is no need to emphasize the dependence on $T$. Since $\xi_{T}$ is log-normally distributed, both $F_\xi$ and $G_\xi$ are $C^{\infty}$ functions. 
The following lemma can be found in \cite{jin2008behavioral}.

\begin{lemma}[\cite{jin2008behavioral}]\label{lemma:3.2} We have
	$\E \left[ \xi_T G_X \left(1-F_{\xi}(\xi_T) \right) \right] \le \E[\xi_T X_T]$ for any lower bounded
	random variable $X_T$ whose quantile function is $G_X$. Furthermore, if $\E [\xi_T G_X(1-F_{\xi}(\xi_T))] < \infty$, then the inequality
	becomes equality if and only if $X_T=G_X \left(1-F_{\xi}(\xi_T) \right), ~a.s.$
\end{lemma}

From Lemmas \ref{lemma:3.1} and \ref{lemma:3.2}, we know that if $X_T$ is optimal to \eqref{prob:martingale}, then $X_T=G_X(1-F_{\xi}(\xi_T))$ where $G_X$ is the quantile function of $X_T$. Let $R_T$ be the log-return of $X_T$. We have
\begin{equation*}
	\E [R_T]=\int_{[0,1)} \frac{1}{T} \ln \frac{ G_{X}(z)}{x} dz,
\end{equation*}
\begin{equation*}
	\rho _{\Phi} (R_T)=-\int _{[0,1)} \frac{1}{T} \ln \frac{ G_{X}(z)}{x} \Phi (dz),
\end{equation*}
and
\begin{equation*}
	\E[\xi_T X_T]=\int_{[0,1)} G_X(z) G_{\xi} (1-z)dz.
\end{equation*}
Therefore, we can consider the following quantile formulation of \eqref{prob:martingale}
\begin{equation}\label{prob:quantile}
\begin{aligned}
\max _{G \in \mathbb{G} } ~ & \lambda \int_{[0,1)} \frac{1}{T} \ln \frac{ G(z)}{x} dz+\int _{[0,1)} \frac{1}{T} \ln \frac{ G(z)}{x} \Phi (dz) \\
\text{subject to} ~ 
& \int_{[0,1)} G(z) G_{\xi} (1-z)dz=x,
\end{aligned}
\end{equation}
where the decision variable $G$ is the quantile function of the terminal wealth. Once we obtain the optimal solution $G^{*}$ to \eqref{prob:quantile}, then the optimal solution to \eqref{prob:martingale} is given by $$X_T^{*}=G^{*} \left( 1-F_{\xi}(\xi_T) \right).$$

Define
\begin{equation}
w(s)=\frac{ \int_{[0,s)} G_{\xi} (1-z)dz }{\E [\xi_T]}, ~ s \in [0,1].
\end{equation}
Because $\xi_{T}$ is log-normally distributed, $w$ is a $C^{\infty}$ function with $w(0)=0$, $w(1)=1$ 
and $w'>0$, $w''<0$ on $(0,1)$. Let $w^{-1}$ be the inverse function of $w$ and define $$H(s)=G \left( w^{-1} (s) \right), ~ s \in [0,1].$$ Then $w^{-1}$ is a $C^{\infty}$ function with $w^{-1}(0)=0$, $w^{-1}(1)=1$, and $(w^{-1})'>0$, $(w^{-1})''>0$ on $(0,1)$. 
It is easy to see $G \in \mathbb{G}$ if and only if $H \in \mathbb{G}$, and $G \in \mathbb{G}_{\Phi}$ if and only if $H \in \mathbb{H}_{\Phi}$, where
\begin{equation*}
	\mathbb{H}_{\Phi}=\Big\{H \in \mathbb{G} \colon H\left( w(s) \right)>0 \text{ if } \Phi ( [0,s] )>0, ~ \forall s \in[0,1] \Big\}.
\end{equation*}
In terms of new notation, we have
\begin{equation*}
	\int_{[0,1)} G(z) G_{\xi} (1-z)dz=\E [\xi_T] \int_{[0,1)} H(s) ds,
\end{equation*}
\begin{equation*}
\int_{[0,1)} \frac{1}{T} \ln \frac{ G(z)}{x} dz=\int_{[0,1)} \frac{1}{T} \ln \frac{ H(s)}{x} dw^{-1} (s) ,
\end{equation*}
and
\begin{equation*}
\int _{[0,1)} \frac{1}{T} \ln \frac{ G(z)}{x} \Phi (dz)=\int _{[0,1)} \frac{1}{T} \ln \frac{ H(s)}{x} d \Phi ([0,w^{-1} (s)]). 
\end{equation*}
Consequently, solving \eqref{prob:martingale} reduces to solving the following quantile optimization problem (after dropping constant terms)
\begin{equation}\label{prob:quantile H}
\begin{aligned}
\max _{H \in \mathbb{G} } ~ &~ \lambda \int_{[0,1)} \ln H(s) dw^{-1} (s)+\int _{[0,1)} \ln H(s) d \Phi ([0,w^{-1} (s)]) \\
\text{subject to} ~ 
&~ \int_{[0,1)} H(s)ds=\frac{x}{\E [\xi_T]}.
\end{aligned}
\end{equation}
This is a concave optimization problem, so it can be tackled by the Lagrange method. 

Define the Lagrangian
\begin{equation*}
L(H(\cdot) ; \lambda, \eta)=\lambda \int_{[0,1)} \ln H(s) dw^{-1} (s)+\int _{[0,1)} \ln H(s) d \Phi ([0,w^{-1} (s)])-\eta \int_{[0,1)}H (s) ds,
\end{equation*}
where $\eta > 0$ is a Lagrange multiplier to fit the budget constraint in \eqref{prob:quantile H}. 
Define
\begin{equation*}
\varphi ( s ; \lambda)=\frac{ \Phi ([0,w^{-1} (s)])+\lambda w^{-1} (s)}{1+\lambda}, ~ s \in [0,1],
\end{equation*}
and its left-continuous version
\begin{equation*}
\varphi ( s-; \lambda)=\frac{ \Phi ([0,w^{-1} (s)))+\lambda w^{-1} (s)}{1+\lambda}, ~ s \in (0,1].
\end{equation*}
We additionally set $\varphi ( 0-; \lambda)=0$. 

We then have
\begin{equation*}
L(H(\cdot) ; \lambda, \eta)=(1+\lambda) \int _{[0,1)} \ln H(s) d\varphi ( s ; \lambda)-\eta \int_{[0,1)} H (s) ds,
\end{equation*}
and we can consider the following optimization problem
\begin{equation}\label{prob:Lagrangian}
\max_{H \in \mathbb{G} }~ L(H(\cdot) ; \lambda, \eta). 
\end{equation}

Inspired by \cite{rogers2009optimal}, \cite{xu2016note}, and \cite{wei2018risk}, we introduce $\delta (s; \lambda)$, the convex envelope function of $\varphi ( s-; \lambda)$ on $[0,1]$, given by 
\begin{equation}\label{eq:concave envelope}
\delta (s; \lambda)=\sup_{0 \le a \le s \le b \le 1} \frac{(b-s)\varphi (a-; \lambda)+(s-a)\varphi (b-; \lambda)}{b-a}, ~s \in [0,1].
\end{equation}
The convex envelope $\delta (s; \lambda)$ is the largest convex function dominated by $\varphi ( s-; \lambda)$, and is affine on the set $\big\{s \in (0,1) \colon \delta (s; \lambda) < \varphi ( s-; \lambda)\big\}.$

The following proposition presents the optimal solution to \eqref{prob:Lagrangian}.
\begin{proposition}\label{prop:3.1}
The optimal solution to \eqref{prob:Lagrangian} is given by $$H^{*} (s; \lambda , \eta)=\frac{1+\lambda}{\eta} \delta' (s; \lambda), ~ s \in [0,1],$$ where $ \delta' (s; \lambda)$ is the right derivative of $\delta (s; \lambda)$ with respect to $s$.
\end{proposition}

We want to find a Lagrange multiplier $\eta$ such that $H^{*} (s; \lambda , \eta)$ satisfies the budget constraint in \eqref{prob:quantile H}. 
Clearly, \begin{equation*}
	 \int_{[0,1)} H^{*} (s; \lambda , \eta)ds=\int_{[0,1)} \frac{1+\lambda}{\eta} \delta' (s; \lambda)ds=\frac{1+\lambda}{\eta}=\frac{x}{\E [\xi_T]}. 
\end{equation*}
and consequently
$$\eta=\frac{1+\lambda}{x}\E [\xi_T].$$

We are ready to state the optimal solution to \eqref{prob:quantile H}. 

\begin{proposition}\label{prop:3.2}
	The optimal solution to \eqref{prob:quantile H} is given by
	\begin{equation*}
		H^{*} \left(s; \lambda , \frac{1+\lambda}{x}\E [\xi_T] \right)=\frac{x}{\E [\xi_T]}\delta' (s; \lambda) .
	\end{equation*}
\end{proposition}

Finally, the optimal solution to \eqref{prob:martingale} is given by
\begin{equation*}
X^{*}_{T,\lambda}=H^{*} \left( w(1-F_{\xi}(\xi_T); \lambda , \frac{1+\lambda}{x}\E [\xi_T] \right)=\frac{x}{\E [\xi_T]}\delta' \left(w(1-F_{\xi}(\xi_T); \lambda \right).
\end{equation*}
In particular, we can obtain the min-risk portfolio by setting $\lambda=0$: 
\begin{equation*}
X_{T,0}^{*}=\frac{x}{\E [\xi_T]} \delta' (w (1-F_{\xi}(\xi_T)); 0),
\end{equation*}
which solves \eqref{prob:min risk}.

We summarize the main results of the paper in the following proposition.

\begin{proposition}[Efficient portfolio]\label{prop:efficient}
	The efficient portfolio, i.e., the optimal solution to \eqref{prob:martingale} is
	\begin{equation*}
	X^{*}_{T,\lambda}=\frac{x}{\E [\xi_T]} \delta' (w (1-F_{\xi}(\xi_T)); \lambda ),
	\end{equation*}
	and the corresponding log-return is $$R^{*}_{T,\lambda}=\frac{1}{T} \ln \left( \frac{ \delta' (w (1-F_{\xi}(\xi_T)); \lambda ) }{\E [\xi_T]} \right).$$
	In particular, 	the min-risk portfolio, i.e., the optimal solution to \eqref{prob:min risk} is
	\begin{equation*}
	X_{T,0}^{*}=\frac{x}{\E [\xi_T]} \delta' (w (1-F_{\xi}(\xi_T)); 0),
	\end{equation*}
	and the corresponding log-return is $$R^{*}_{T,0}=\frac{1}{T} \ln \left( \frac{ \delta' (w (1-F_{\xi}(\xi_T));0 ) }{\E [\xi_T]} \right).$$
\end{proposition}
 
\section{Examples with explicit solution}\label{sec:examples}
In this section, we present two examples to illustrate our general results. In particular, we consider the optimization problem \eqref{prob:martingale} when the WVaR risk measure is given by either VaR or ES, two popular risk measures.

\subsection{Mean-VaR efficient portfolio}
In this subsection, we specialize our setting to the mean-VaR optimization problem. In particular, we consider the optimization problem \eqref{prob:martingale} when the WVaR risk measure is given by the VaR at a confidence level $0<\alpha<1$, namely
\begin{equation*}
\rho _{\Phi} (X)=\text{VaR}_{\alpha}(X)=-G_{X}(\alpha).
\end{equation*}
In other words, $\Phi$ is given by the Dirac measure at $\alpha$.

\begin{proposition}\label{prop:4.1}
When $\Phi$ is given by the Dirac measure at $\alpha\in (0,1)$, we have the following assertions.
	\begin{description}
\item[Case $\lambda=0$.] 
	\begin{enumerate}
		\item The minimum-VaR (min-VaR) terminal wealth is 
		\begin{equation*}
		X^{\text{VaR}}_{T,0}=
		\begin{cases}
		0, ~ & \xi_T > \xi_{\alpha},\\
		\underline{X}_{\text{VaR}} , ~ & \xi_T \le \xi_{\alpha},
		\end{cases}
		\end{equation*}
		where 
		\begin{equation*}
			\begin{aligned}
				\underline{X}_{\text{VaR}} &=\frac{x}{\E [\xi_T]} \cdot \frac{ 1}{ 1-w(\alpha)} ,\quad 
				\xi_{\alpha}=G_{\xi} (1-\alpha).
			\end{aligned}
		\end{equation*}
		
		\item The optimal log-return is 
		$$R^{\text{VaR}}_{T,0}=\frac{1}{T} \ln \frac{X^{\text{VaR}}_{T,0}}{x}.$$

		\item The expected optimal log-return is $\E [R^{\text{VaR}}_{T,0}]=-\infty.$
		
		\item The VaR of the optimal log-return is 
		$$\text{VaR}_{\alpha}(R^{\text{VaR}}_{T,0})=\frac{1}{T} \ln \frac{\underline{X}_{\text{VaR}}}{x}.$$

	\end{enumerate}

\item[Case $0<\lambda<\infty$.] 
	\begin{enumerate}
		\item The mean-VaR efficient terminal wealth is 
		\begin{equation*}
		X^{\text{VaR}}_{T,\lambda}=
		\begin{cases}
		\frac{\lambda}{1+\lambda} \cdot \frac{x}{\xi_T} , ~ & \xi_T > \xi_{\alpha},\\
		 \underline{X}_{\text{VaR}} , ~ & \underline{\xi}_{\text{VaR}} < \xi_T \le \xi_{\alpha} ,\\
		 \frac{\lambda}{1+\lambda} \cdot \frac{x}{\xi_T} , ~ & \xi_T \le \underline{\xi}_{\text{VaR}},
		\end{cases}
		\end{equation*}
		where 
		\begin{equation*}
		\begin{aligned}
		\underline{X}_{\text{VaR}} &=\frac{\lambda}{1+\lambda} \cdot \frac{x}{ \underline{\xi}_{\text{VaR}}} ,\\
		\xi_{\alpha} &=G_{\xi} (1-\alpha),\\
		\underline{\xi}_{\text{VaR}} &=G_{\xi} (1-w^{-1}(s^{*}(\lambda) )) ,
		\end{aligned}
		\end{equation*}
		and $s^{*}(\lambda)$ is given in Lemma \ref{lemma: f1}.

		\item The optimal log-return is 
		$$R^{\text{VaR}}_{T,\lambda}=\frac{1}{T} \ln \frac{X^{\text{VaR}}_{T,\lambda}}{x} .$$

		\item The VaR of the optimal log-return is 
		$$\text{VaR}_{\alpha}(R^{\text{VaR}}_{T,\lambda})=\frac{1}{T} \ln \frac{ \underline{X}_{\text{VaR}} }{x} .$$

	\end{enumerate}
\end{description}

\end{proposition}

Figure \ref{figure:VaR wealth1} depicts the optimal terminal payoff of the min-VaR portfolio ($\lambda=0$), which resembles a digital option. Essentially, the investor invests all the money in a digital option that pays $\underline{X}_{\text{VaR}}$ in the good states of the market ($\xi < \xi_{\alpha}$) and 0 otherwise. The probability of winning the option depends solely on the confidence level of VaR and is given by $\P (\xi \le \xi_{\alpha})=1-\alpha$. 

\begin{figure}[H]
	\centering
		\centering
		\includegraphics[width=0.6\textwidth]{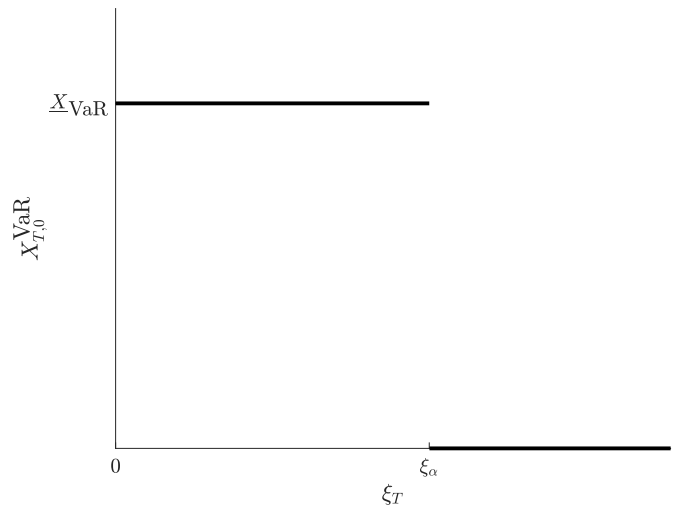}
	\caption{The Min-VaR Efficient Terminal Wealth ($\lambda=0$)}\label{figure:VaR wealth1}
\end{figure}

Figure \ref{figure:VaR wealth2} displays the optimal terminal payoff of the mean-VaR efficient portfolio ($0<\lambda<\infty$). The investor classifies market scenarios into three subsets: in the good states ($\xi \le \underline{\xi}_{\text{VaR}}$) and in the bad states ($\xi > \xi_{\alpha}$), the terminal payoff is a fraction ($\lambda/(1+\lambda)$) of the growth optimal portfolio; in the intermediate states ($\underline{\xi}_{\text{VaR}} < \xi \le \xi_{\alpha}$), the investor receives a constant payoff $\underline{X}_{\text{VaR}}$. Moreover, the terminal wealth has a jump discontinuity at $\xi=\xi_{\alpha}$ and the corresponding log-returns are always finite (but can be extremely large or small).

\begin{figure}[H]
		\centering
		\includegraphics[width=0.6\textwidth]{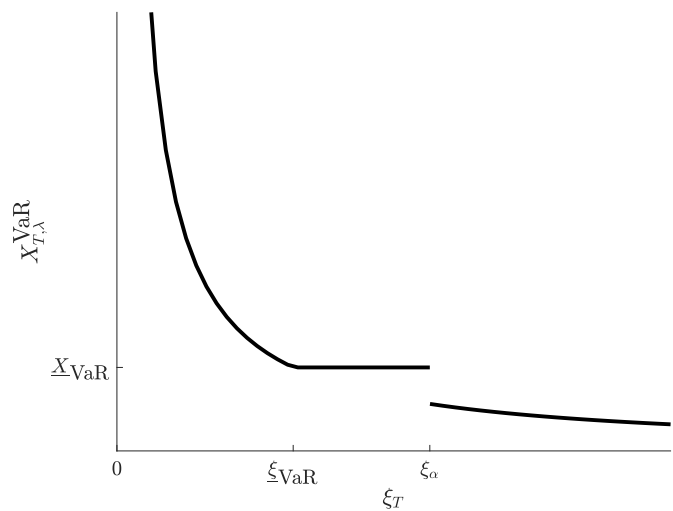}
	\caption{The Mean-VaR Efficient Terminal Wealth ($0<\lambda<\infty$)}\label{figure:VaR wealth2}
\end{figure}

Figure \ref{figure:VaR efficient} plots the mean-VaR efficient frontiers for different confidence levels $\alpha$. The efficient frontier is a concave curve that connects the growth optimal portfolio (colored dots) with the min-VaR portfolio (not shown in the graph). The growth optimal portfolio has the highest expected log-return but also the highest VaR. By contrast, the min-VaR portfolio has the smallest expected log-return (negative infinity) but also the lowest VaR. Figure \ref{figure:VaR efficient} also displays a sensitivity analysis of the efficient frontier with respect to $\alpha$, the confidence level of VaR. As $\alpha$ increases, the efficient frontier shifts to the left: for a given level of the expected log-return, the VaR of the corresponding efficient portfolio decreases as $\alpha$ increases.

\begin{figure}[H]
		\centering
		\includegraphics[width=0.6\textwidth]{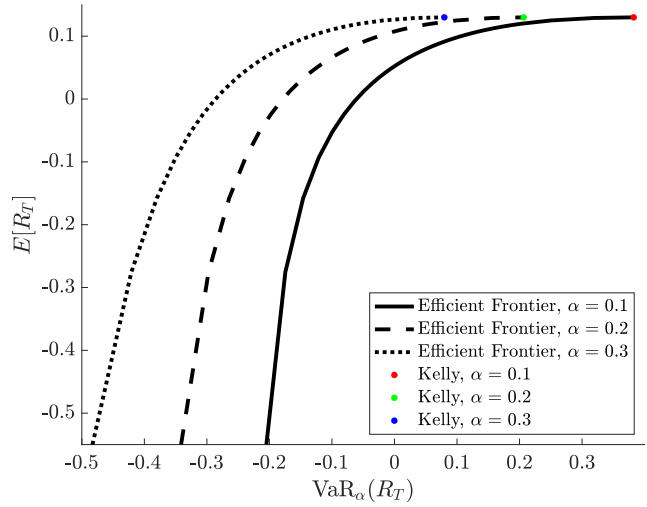}\\
	
	\caption{The Mean-VaR Efficient Frontier with $0<\lambda<\infty$, $T=1$, $r=0.05$, and $\theta=0.4$ }\label{figure:VaR efficient}
\end{figure}

As the optimal terminal wealth is known, we can solve for the optimal time-$t$ wealth and portfolio policy.

\begin{corollary}\label{coro:4.1}
We have the following assertions.

\begin{description}
\item[Case $\lambda=0$.] 
\begin{enumerate}
\item The min-VaR efficient wealth at time $t$ is 
\begin{equation*}
X^{\text{VaR}}_{t,0}=\underline{X}_{\text{VaR}} e^{-r (T-t)} N \left( d_2 ( t, \xi_t, \xi_{\alpha}) \right).
\end{equation*}

\item The optimal portfolio policy at time $t$ is 
\begin{equation*}
\pi^{\text{VaR}}_{t,0}=\frac{\underline{X}_{\text{VaR}} e^{-r (T-t)} \nu \left( d_2 (t, \xi_t, \xi_{\alpha}) \right)}{ \sigma \sqrt{T-t} } .
\end{equation*}
\end{enumerate}

\item[Case $0<\lambda<\infty$.] 
\begin{enumerate}
\item The mean-VaR efficient wealth at time $t$ is 
\begin{equation*}
\begin{aligned}
X^{\text{VaR}}_{t,\lambda} 
=&\frac{\lambda}{1+\lambda} \cdot \frac{x}{\xi_t} \left( N \left(-d_1 (t, \xi_t, \xi_{\alpha}) \right)+N \left( d_1 (t, \xi_t, \underline{\xi}_{\text{VaR}} ) \right) \right) \\
&+\underline{X}_{\text{VaR}} e^{-r(T-t)} \left( N \left( d_2 (t, \xi_t, \xi_{\alpha}) \right)-N \left( d_2 (t, \xi_t, \underline{\xi}_{\text{VaR}} ) \right) \right).
\end{aligned}
\end{equation*}

\item The optimal portfolio policy at time $t$ is 
\begin{equation*}
\begin{aligned}
\pi^{\text{VaR}}_{t,\lambda}=& \frac{\lambda}{1+\lambda} \cdot \frac{x}{\xi_t} \cdot \left( N\left(-d_1(t, \xi_t, \xi_{\alpha}) \right)+N\left(d_1(t, \xi_t, \underline{\xi}_{\text{VaR}}) \right) \right) \frac{\theta}{\sigma} \\
&+\frac{ e^{-r (T-t)} \nu \left( d_2 (t, \xi_t, \xi_{\alpha}) \right) }{\sigma \sqrt{T-t} } \cdot \left( \underline{X}_{\text{VaR}}-\frac{\lambda}{1+\lambda} \cdot \frac{x}{\xi_{\alpha}} \right).
\end{aligned}
\end{equation*}
\end{enumerate}
\end{description}
Here and hereafter 
\begin{equation*}
\begin{aligned}
d_1 (t, \xi_t, y) &=\frac{\ln \frac{y}{\xi _t}+(r+\frac{\theta ^2 }{2} )(T-t) }{\theta \sqrt{T-t}}, \\
d_2 (t, \xi_t, y) &=d_1 (t, \xi_t, y)-\theta \sqrt{T-t},
\end{aligned}
\end{equation*}
and $N (\cdot)$ is the standard normal distribution function, and $\nu (\cdot)$ is the standard normal probability density function. 
\end{corollary}

\subsection{Mean-ES efficient portfolio}

In this subsection, we specialize our setting to the mean-ES optimization problem. In particular, we consider the optimization problem \eqref{prob:martingale} when the WVaR risk measure is given by the ES at a confidence level $0<\alpha<1$, namely 
\begin{equation*}
\rho _{\Phi} (X)=\text{ES}_{\alpha} (X)=-\frac{1}{\alpha} \int_0^{\alpha} G_{X} (z)dz.
\end{equation*}
In other words, $\Phi$ admits a density $\phi (z)=\frac{1}{\alpha}\mathbf{1}_{z \le \alpha}$, for all $z \in [0,1]$.

\begin{proposition}\label{prop:4.2}
When $\Phi$ admits a density $\phi (z)=\frac{1}{\alpha}\mathbf{1}_{z \le \alpha}$ with $0<\alpha< 1$, we have the following assertions.
	
\begin{description}
\item[Case $\lambda=0$.] 	\begin{enumerate}
		\item The min-ES efficient terminal wealth is 
		\begin{equation*}
		X^{\text{ES}}_{T,0}=
		\left \{
		\begin{aligned}
		& \frac{x}{\alpha \xi_T} , ~ & \xi_T > \overline{\xi}_{\text{ES}},\\
		& \underline{X}_{\text{ES}}, ~ & \xi_T \le \overline{\xi}_{\text{ES}},
		\end{aligned}
		\right.
		\end{equation*}
		where 
		\begin{equation*}
		\begin{aligned}
		\underline{X}_{\text{ES}} &=\frac{x}{\alpha \overline{\xi}_{\text{ES}} },\\
		\overline{\xi}_{\text{ES}} &=G_{\xi} (1-w^{-1}(t_0)),
		\end{aligned}
		\end{equation*}
		and $t_0$ is given in Lemma \ref{lemma: f2}. 
		
		\item The optimal log-return is 
		$$R^{\text{ES}}_{T,0}=\frac{1}{T} \ln \frac{X^{\text{ES}}_{T,0}}{x}.$$
		\item The ES of the optimal log-return is 
\begin{align*}
\text{ES}_{\alpha}(R^{\text{ES}}_{T,0})&=\frac{1}{\alpha T} \left[ \ln \alpha \cdot N \left( - \frac{\ln \xi_{\alpha} + \left( r + \frac{\theta^2}{2} \right)T}{ \theta \sqrt{T}} \right) \right. \\
&\quad\;+ \ln \overline{\xi}_{\text{ES}} \cdot \left( N \left( \frac{\ln \overline{\xi}_{\text{ES}} + \left( r + \frac{\theta^2}{2} \right)T}{ \theta \sqrt{T}} \right) - N \left( \frac{\ln \xi_{\alpha} + \left( r + \frac{\theta^2}{2} \right)T}{ \theta \sqrt{T}} \right) \right) \\
&\quad\;\left. + \frac{\theta \sqrt{T}}{\sqrt{2 \pi}} e^{ - \frac{\left( \ln \overline{\xi}_{\text{ES}} + \left( r + \frac{\theta^2}{2} \right)T \right)^2}{2 \theta^2 T}} - \left( r + \frac{\theta^2}{2} \right) T N \left( - \frac{\ln \overline{\xi}_{\text{ES}} + \left( r + \frac{\theta^2}{2} \right)T}{ \theta \sqrt{T}} \right) \right] ,
\end{align*}
where $\xi_{\alpha} =G_{\xi} (1-\alpha).$
	\end{enumerate}
	
\item[Case $0<\lambda<\infty$.] 
	
	\begin{enumerate}
		\item The mean-ES efficient terminal wealth is 
		\begin{equation*}
		X^{\text{ES}}_{T,\lambda}=
		\begin{cases}
		 \frac{\frac{1}{\alpha}+\lambda}{1+\lambda} \cdot \frac{x}{\xi_T} , ~ & \xi_T > 	\overline{\xi}_{\text{ES}},\\
		 \underline{X}_{\text{ES}} , ~ & \underline{\xi}_{\text{ES}} < \xi_T \le 	\overline{\xi}_{\text{ES}} ,\\
		\frac{\lambda}{1+\lambda} \cdot \frac{x}{\xi_T} , ~ & \xi_T \le \underline{\xi}_{\text{ES}}, 
		\end{cases}
		\end{equation*}
		where 
		\begin{equation*}
		\begin{aligned}
		\underline{X}_{\text{ES}} &=\frac{\frac{1}{\alpha}+\lambda}{1+\lambda} \cdot \frac{x}{\overline{\xi}_{\text{ES}} }=\frac{\lambda}{1+\lambda} \cdot \frac{x}{\underline{\xi}_{\text{ES}}},\\
		\overline{\xi}_{\text{ES}} &=G_{\xi} (1-w^{-1}( t_1(\lambda) )),\\
		\underline{\xi}_{\text{ES}} &=\frac{\lambda}{\frac{1}{\alpha}+\alpha} \overline{\xi}_{\text{ES}},
		\end{aligned}
		\end{equation*}
		and $t_1 (\lambda)$ is given in Lemma \ref{lemma: f3}.

		\item The optimal log-return is 
		$$R^{\text{ES}}_{T,\lambda}=\frac{1}{T} \ln \frac{X^{\text{ES}}_{T,\lambda}}{x} .$$
		
		\item The ES of the optimal log-return is 
\begin{align*}
\text{ES}_{\alpha}(R^{\text{ES}}_{T,\lambda})&=\frac{1}{\alpha T} \left[ \ln \left( \frac{1+\lambda}{\frac{1}{\alpha} + \lambda} \right) \cdot N \left( - \frac{\ln \xi_{\alpha} + \left( r + \frac{\theta^2}{2} \right)T}{ \theta \sqrt{T}} \right) \right. \\
&\quad\;+ \ln \overline{\xi}_{\text{ES}} \cdot \left( N \left( \frac{\ln \overline{\xi}_{\text{ES}} + \left( r + \frac{\theta^2}{2} \right)T}{ \theta \sqrt{T}} \right) - N \left( \frac{\ln \xi_{\alpha} + \left( r + \frac{\theta^2}{2} \right)T}{ \theta \sqrt{T}} \right) \right) \\
&\quad\;\left. + \frac{\theta \sqrt{T}}{\sqrt{2 \pi}} e^{ - \frac{\left( \ln \overline{\xi}_{\text{ES}} + \left( r + \frac{\theta^2}{2} \right)T \right)^2}{2 \theta^2 T}} - \left( r + \frac{\theta^2}{2} \right) T N \left( - \frac{\ln \overline{\xi}_{\text{ES}} + \left( r + \frac{\theta^2}{2} \right)T}{ \theta \sqrt{T}} \right) \right], 
\end{align*}
where $\xi_{\alpha} =G_{\xi} (1-\alpha).$
	\end{enumerate}
\end{description}
\end{proposition}

Figure \ref{figure:ES wealth1} depicts the optimal terminal payoff of the min-ES portfolio ($\lambda=0$). The investor classifies market scenarios into two subsets: in the good states ($\xi \le \overline{\xi}_{\text{ES}}$), the investor receives a constant payoff $\underline{X}_{\text{ES}}$; in the bad states ($\xi > \overline{\xi}_{\text{ES}}$), the payoff is a multiple ($1/\alpha$) of the growth optimal portfolio. 

\begin{figure}[H]
	\centering
		\centering
		\includegraphics[width=0.6\textwidth]{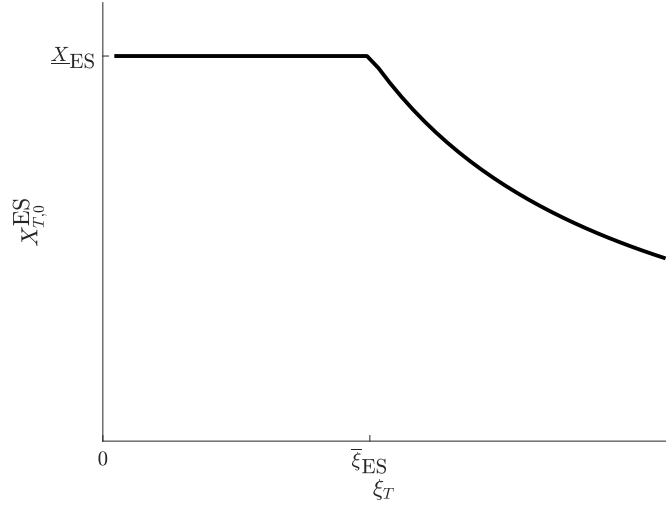}
	\caption{The Min-ES Efficient Terminal Wealth ($\lambda=0$)}\label{figure:ES wealth1}
\end{figure}

Figure \ref{figure:ES wealth2} displays the optimal terminal payoff of the mean-ES efficient portfolio ($0<\lambda<\infty$). The investor classifies market scenarios into three subsets: in the good states ($\xi \le \underline{\xi}_{\text{ES}}$), the terminal payoff is a fraction ($\lambda/(1+\lambda)$) of the growth optimal portfolio; in the intermediate states ($\underline{\xi}_{\text{ES}} < \xi \le \overline{\xi}_{\text{ES}}$), the investor receives a constant payoff $\underline{X}_{\text{ES}}$; in the bad states ($\xi > \overline{\xi}_{\text{ES}}$), the terminal payoff is a multiple ($(\frac{1}{\alpha}+\lambda)/(1+\lambda)$) of the growth optimal portfolio. In contrast to the mean-VaR efficient portfolio, the terminal payoff of the mean-ES efficient portfolio is continuous in the state price density.

\begin{figure}[H]
	\centering
		\centering
		\includegraphics[width=0.6\textwidth]{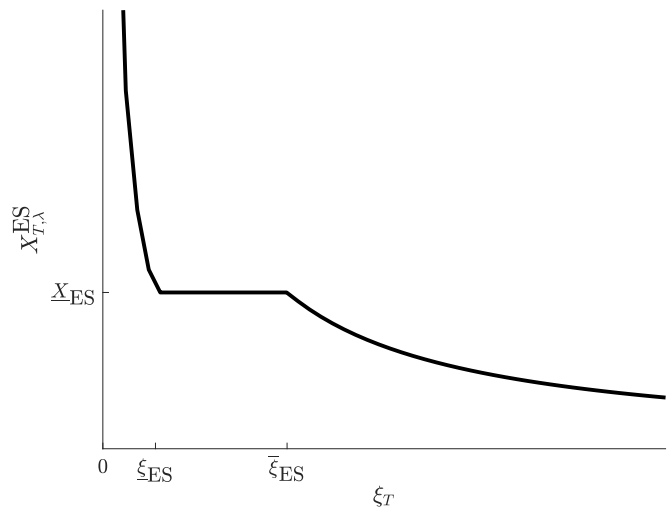}
	\caption{The Mean-ES Efficient Terminal Wealth ($0<\lambda<\infty$)}\label{figure:ES wealth2}
\end{figure}

Figure \ref{figure:ES efficient} shows the mean-ES efficient frontiers for different confidence levels $\alpha$. The efficient frontier is a concave curve that connects the growth optimal portfolio (colored dots) with the min-ES portfolio (colored crosses). The growth optimal portfolio has the highest expected log-return but also the highest ES. By contrast, the min-ES portfolio has the smallest expected log-return but also the lowest ES. In contrast to the min-VaR portfolio, the risk of the min-ES portfolio is finite and thus the mean-ES efficient frontier is a finite curve. Figure \ref{figure:ES efficient} also displays a sensitivity analysis of the efficient frontier with respect to $\alpha$, the confidence level of ES. As $\alpha$ increases, the efficient frontier shifts to the left: for a given level of the expected log-return, the ES of the corresponding efficient portfolio decreases as $\alpha$ increases. In particular, the minimum ES that the investor can achieve is decreasing in $\alpha$.

\begin{figure}[H]
	\centering
	\includegraphics[width=0.6\textwidth]{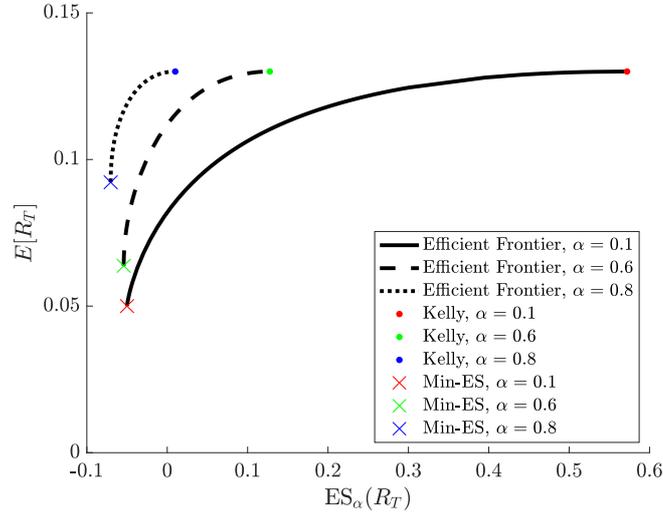}
	\caption{The Mean-ES Efficient Frontier with $0<\lambda<\infty$, $T=1$, $r=0.05$, and $\theta=0.4$. }\label{figure:ES efficient}
\end{figure}

The following corollary presents the optimal time-$t$ wealth and portfolio policy. The proof is similar to that of Corollary \ref{coro:4.1} and thus we omit it.

\begin{corollary}
We have the following assertions.

\begin{description}
\item[Case $\lambda=0$.] \begin{enumerate}
\item The min-ES efficient wealth at time $t$ is 
\begin{equation*}
X^{\text{ES}}_{t,0}=\frac{x}{\alpha \xi_t} N \left(-d_1 (t, \xi_t, \overline{\xi}_{\text{ES}} ) \right)+\underline{X}_{\text{ES}} e^{-r (T-t)} N \left( d_2 (t, \xi_t, \overline{\xi}_{\text{ES}} ) \right).
\end{equation*}

\item The efficient portfolio policy at time $t$ is 
\begin{equation*}
\pi^{\text{ES}}_{t,0}=\frac{x}{\alpha \xi_t} N \left(-d_1 (t, \xi_t, \overline{\xi}_{\text{ES}} ) \right) \frac{\theta}{\sigma} .
\end{equation*}
\end{enumerate}

\item[Case $0<\lambda<\infty$.] 
\begin{enumerate}
\item The mean-ES efficient wealth at time $t$ is 
\begin{align*}
X^{\text{ES}}_{t,\lambda} 
&= \frac{\frac{1}{\alpha}+\lambda}{1+\lambda} \cdot \frac{x}{\xi_t} N \left(-d_1 ( t, \xi_t, \overline{\xi}_{\text{ES}} ) \right)\\
&\quad\;+\underline{X}_{\text{ES}} e^{-r(T-t)} \left( N \left( d_2 (t, \xi_t, \overline{\xi}_{\text{ES}} ) \right)-N \left( d_2 (t, \xi_t, \underline{\xi}_{\text{ES}} ) \right) \right)\\
&\quad\;+\frac{\lambda}{1+\lambda} \cdot \frac{x}{\xi_t} N \left( d_1 (t, \xi_t, \underline{\xi}_{\text{ES}} ) \right). 
\end{align*}

\item The efficient portfolio policy at time $t$ is 
\begin{align*}
\pi^{\text{ES}}_{t,\lambda}& =\frac{x}{\xi_t} \cdot \frac{\theta}{\sigma} \left(\frac{\frac{1}{\alpha}+\lambda}{1+\lambda} N \left(-d_1 (t, \xi_t, \overline{\xi}_{\text{ES}} ) \right)+\frac{\lambda}{1+\lambda} N \left( d_1 (t, \xi_t, \underline{\xi}_{\text{ES}} ) \right) \right). 
\end{align*}
\end{enumerate}
\end{description}

\end{corollary}
 
\subsection{Comparison with \cite{he2015dynamic}}
\cite{he2015dynamic} consider a continuous-time mean-risk portfolio choice problem in which the risk is measured by WVaR. They assume the decision-maker minimizes the risk of terminal wealth, while maintaining the expected terminal wealth above a prescribed target. They find that the model can lead to extreme risk-taking behaviors. When bankruptcy is allowed, the optimal terminal wealth is binary, i.e., the investor invests a small amount of money in an extremely risky digital option and saves the rest of the money in the risk-free asset. When bankruptcy is prohibited, the terminal wealth can be three-valued and the optimal strategy is to invest a small amount of money in an extremely risky digital option and put the rest in an asset with moderate risk. These strategies are not commonly seen in practice and are not appropriate for many investors. Furthermore, the optimal value (the risk) is independent of the expected terminal wealth target. Therefore the efficient frontier is a vertical line in the mean-risk plane and there is no explicit trade-off between risk and return. They conclude that using the WVaR on terminal wealth is not an appropriate model of risk for portfolio choice.

In contrast to \cite{he2015dynamic}, our model uses the expected target and risk measure on log-returns instead of terminal wealth. When the risk is evaluated by the VaR or ES, two popular risk measures, we find that the investor classifies market scenarios into different states, in which the terminal payoff is a multiple or fraction of the growth optimal portfolio, or constant. Furthermore, the efficient frontier is a concave curve that connects the min-risk portfolio with the growth optimal portfolio. Our model allows for an explicit characterization of the risk-return trade-off and may serve as a guideline for investors to set reasonable investment targets. Our results demonstrate that it is more appropriate to use the WVaR, in particular, the VaR and ES, on the log-return instead of the terminal wealth for portfolio choice.

\section{Conclusion}\label{sec:conclusion}
We have proposed and solved a dynamic mean-WVaR portfolio choice problem with risk measured to log-returns, as opposed to terminal wealth in \cite{he2015dynamic}. Our model conquers the ill-posedness of the mean-WVaR criterion for terminal wealth in \cite{he2015dynamic}, and allows for an explicit and meaningful characterization of the trade-off between return and risk. We have demonstrated that our proposed mean-WVaR criterion for log-returns is more appropriate and tractable than the mean-WVaR criterion for terminal wealth in serving as a guideline for dynamic portfolio choice.

\newpage

\appendix
\section{Useful Lemmas}\label{sec:A1}

The following function 
\begin{equation}\label{varphi}
\varphi ( s-; \lambda)= 
\begin{cases}
\frac{ \lambda w^{-1} (s)}{1+\lambda}, ~ & 0 \le s \le w(\alpha),\\
\frac{ 1+\lambda w^{-1} (s)}{1+\lambda}, ~ & w(\alpha) < s \le 1, 
\end{cases}
\end{equation}	
is increasing, continuous and convex on $[0, w(\alpha)]$ and on $(w(\alpha),1]$, respectively. It has an upward jump at $s=w(\alpha)$. 
\begin{lemma}\label{lemma: f0} 
If $\varphi$ is defined by \eqref{varphi} with $\lambda=0$. Then its convex envelope is given by 
\begin{equation*}
\delta ( s ; 0)=
\begin{cases}
0, ~ & 0 \le s \le w(\alpha),\\
\frac{ s-w(\alpha) }{ 1-w(\alpha)} , ~ & w(\alpha) < s \le 1. 
\end{cases}
\end{equation*}
\end{lemma}
\begin{proof}
When $\lambda=0$, we have
\begin{equation*}
\varphi ( s-; 0)=
\left \{
\begin{aligned}
&0, ~ & 0 \le s \le w(\alpha),\\
&1 , ~ & w(\alpha) < s \le 1,
\end{aligned}
\right.
\end{equation*} 
It is straightforward to verify $\delta ( s ; 0) $ is convex. 
For any convex function $\psi(s)$ dominated by $\varphi (s-;0)$, 
it is clearly dominated on $[0, w(\alpha)]$ by $\delta (s;0)$ as $\varphi ( s-; 0)=\delta (s;0)$. Because 
$\delta ( w(\alpha) ; 0)=\varphi (w(\alpha)-; 0) \ge \psi (w(\alpha))$ 
and $\delta (1 ; 0)=\varphi (1-; 0) \ge \psi (1)$, 
for any $s \in (w(\alpha),1]$, we have 
\begin{align*}
\delta ( s ; 0) &=\frac{1-s}{1-w(\alpha)} \delta ( w(\alpha) ; 0)+\frac{s-w(\alpha)}{1-w(\alpha)} \delta ( 1 ; 0)\\
&=\frac{1-s}{1-w(\alpha)} \varphi ( w(\alpha)-; 0)+\frac{s-w(\alpha)}{1-w(\alpha)} \varphi ( 1-; 0)\\
&\ge \frac{1-s}{1-w(\alpha)} \psi (w(\alpha))+\frac{s-w(\alpha)}{1-w(\alpha)} \psi (1) \ge \psi (s).
\end{align*}
Therefore, $\delta ( s ; 0)$ is the largest convex function dominated by $\varphi ( s-; 0)$ on $[0,1]$, and thus is the convex envelope of $\varphi ( s-; 0)$. The left panel of Figure \eqref{figure:VaR varphi delta} gives a graphical illustration. 
\end{proof}

\begin{figure}[H]
\centering
\begin{minipage}[t]{3in}
\centering
\includegraphics[width=3in]{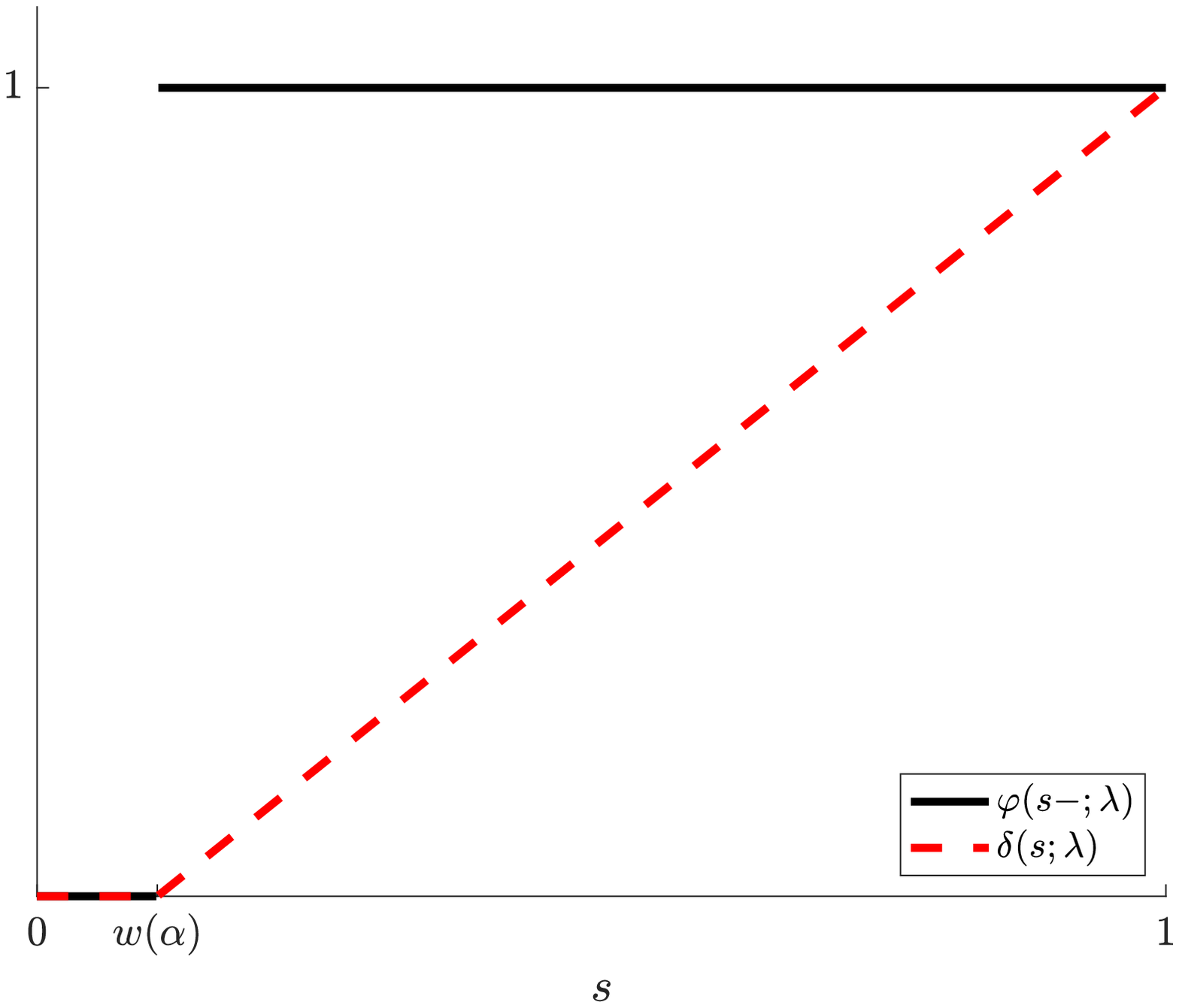}
\centerline{$\lambda=0$}
\end{minipage}
\begin{minipage}[t]{3in}
\centering
\includegraphics[width=3in]{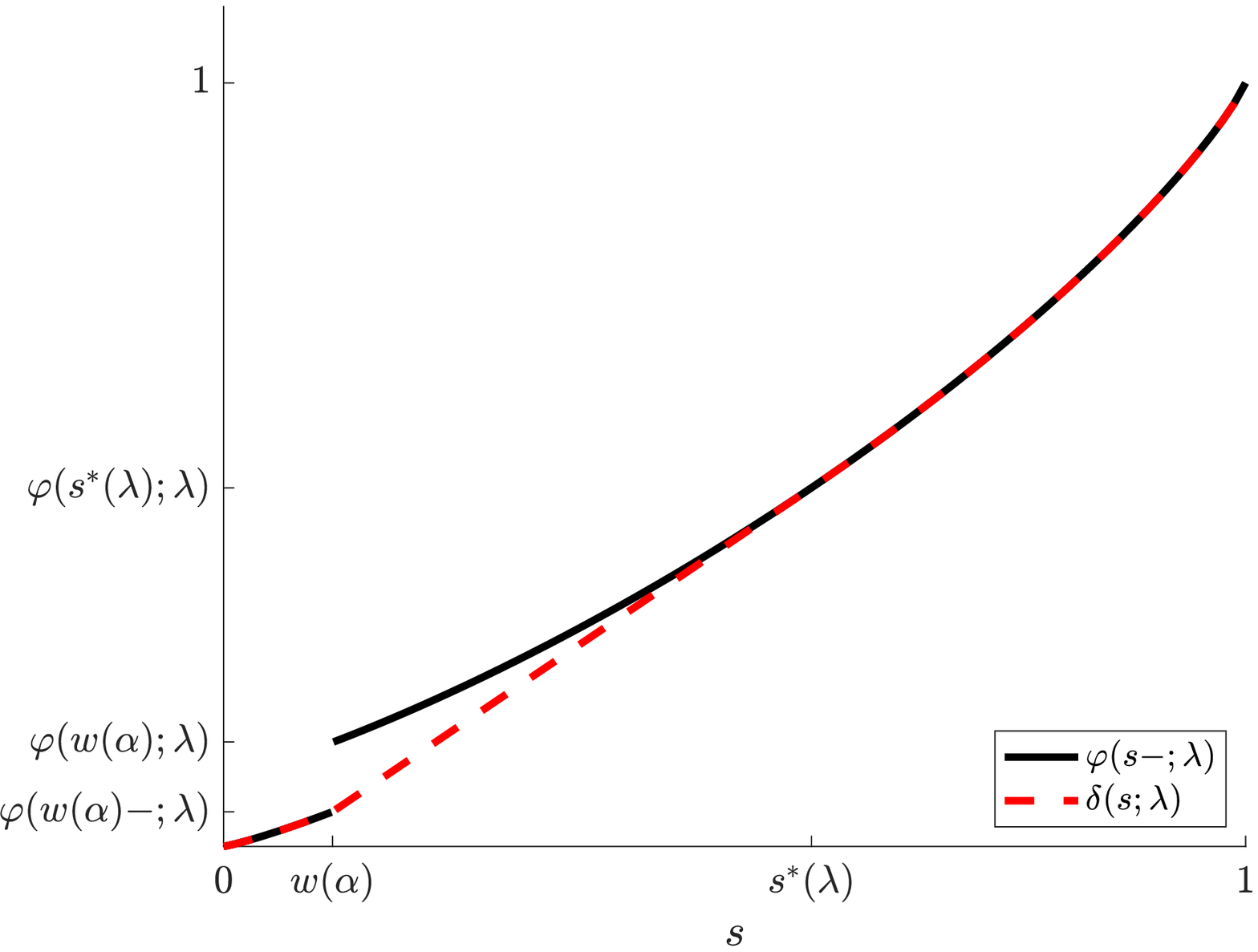}
\centerline{$0<\lambda<\infty$}\bigskip
\end{minipage}

\caption{$\varphi (s-,\lambda)$ and $\delta (s,\lambda)$ for VaR }\label{figure:VaR varphi delta}
\end{figure}

\begin{lemma}\label{lemma: f1} 
If $\varphi$ is defined by \eqref{varphi} with $0<\lambda<\infty$. Then its convex envelope is given by 
\begin{equation*}
\delta ( s ; \lambda)=
\begin{cases}
\varphi ( s-; \lambda), ~ & 0 \le s \le w(\alpha),\\
\varphi ( w(\alpha)-; \lambda)+\varphi'( s^{*}(\lambda) ; \lambda) (s-w(\alpha) ), ~ & w(\alpha) < s \le s^{*}(\lambda),\\
\varphi ( s-; \lambda) , ~ & s^{*}(\lambda) < s \le 1,
\end{cases}
\end{equation*} 
where $s^{*}(\lambda)$ is the unique number in $(w(\alpha),1)$ such that 
	\begin{equation*}
	\varphi'( s^{*}(\lambda)-; \lambda)=	\frac{ \varphi ( s^{*}(\lambda)-; \lambda)-\varphi ( w(\alpha)-; \lambda) }{ s^{*}(\lambda)-w(\alpha)}.
	\end{equation*} 
\end{lemma}
\begin{proof} 
	Let
	$$f_1 (s; \lambda)=\varphi ( s-; \lambda)-\varphi ( w(\alpha)-; \lambda)-\varphi'( s ; \lambda) ( s-w(\alpha) ), ~ w(\alpha)<s<1,$$
	which is continuous in $s$. 
		For $w(\alpha)<s_1<s_2<1$, by integration by parts we have
	\begin{equation*}
	\begin{aligned}
		&\;f_1 (s_1; \lambda)-f_1 (s_2; \lambda) \\
		=&\; \varphi ( s_1-; \lambda)-\varphi ( s_2-; \lambda)-\varphi'( s_1 ; \lambda) ( s_1-w(\alpha) )+\varphi'( s_2 ; \lambda) ( s_2-w(\alpha) ) \\
		=&-\int_{s_1-}^{s_2-} d\varphi( z ; \lambda)+\int_{s_1}^{s_2} d\big(z\varphi'( z ; \lambda) \big)-w(\alpha) \left( \varphi'( s_2 ; \lambda)-\varphi'( s_1 ; \lambda) \right)\\
		=&-\int_{s_1}^{s_2} \varphi'( z-; \lambda) dz+\int_{s_1}^{s_2} z d\varphi'( z ; \lambda)+\int_{s_1}^{s_2} \varphi'( z ; \lambda) dz-w(\alpha) \int_{s_1}^{s_2} d \varphi'( z ; \lambda) \\
		=& \int_{s_1}^{s_2} ( z-w(\alpha) ) d\varphi'( z ; \lambda) >0,
	\end{aligned}
	\end{equation*}
thanks to the convexity of $\varphi( z ; \lambda)$ on $(w(\alpha),1]$, 
	so $f_1 (s; \lambda)$ is a strictly decreasing function of $s$. 
	Meanwhile, 
	\begin{align*}
	\lim_{s \downarrow w(\alpha)} f_1 (s; \lambda)
	&=\lim_{s \downarrow w(\alpha)} \frac{1+\lambda (w^{-1} (s)-\alpha)-\lambda (w^{-1} (s))'( s-w(\alpha))}{1+\lambda}=\frac{1}{1+\lambda},
	\end{align*}
	and 
		\begin{equation*}
	\lim_{s \uparrow 1} f_1 (s; \lambda)=\varphi ( 1 ; \lambda)-\varphi ( w(\alpha)-; \lambda)-( 1-w(\alpha) ) \lim_{s \uparrow 1}\varphi'( s ; \lambda)=-\infty.
	\end{equation*}
Therefore, there exists a unique $s^{*}(\lambda) \in (w(\alpha),1)$ such that $f_1 ( s^{*}(\lambda) ; \lambda)=0$. 

Thanks to the convexity of $\varphi( s ; \lambda)$ on $[0,w(\alpha)]$ and on $(w(\alpha),1]$, it is easy to verify that $\delta ( s ; \lambda)$ is a continuous convex function dominated by $\varphi ( s-; \lambda)$ on $[0,1]$. 
For any convex function $\psi(s)$ dominated by $\varphi (s-;\lambda)$, 
it is clearly dominated on $[0, w(\alpha)]\cup(s^{*}(\lambda),1]$ by $\delta (s;\lambda)$ as 
$\delta (s;\lambda)=\varphi (s-;\lambda)$. 
By the same argument as in the proof of Lemma \ref{lemma: f0}, we have that $\psi(s)$ is dominated by $\delta (s;\lambda)$ on $(w(\alpha), s^{*}(\lambda)]$ too. 
Therefore, $\delta ( s ; \lambda)$ is the largest convex function dominated by $\varphi ( s-; \lambda)$, and thus the convex envelope of $\varphi ( s-; \lambda)$. The right panel of Figure \eqref{figure:VaR varphi delta} gives a graphical illustration. 
\end{proof}

The following function 
\begin{equation}\label{varphi2}
\varphi ( s-; \lambda)=
\begin{cases}
\frac{ (\frac{1}{\alpha}+\lambda )w^{-1} (s)}{1+\lambda}, ~ & 0 \le s \le w(\alpha)\\
\frac{ 1+\lambda w^{-1} (s)}{1+\lambda}, ~ & w(\alpha) < s \le 1.
\end{cases}
\end{equation} 
is increasing, and convex on $[0, w(\alpha)]$ and on $(w(\alpha),1]$, respectively. It is continuous on $[0,1]$.

\begin{lemma}\label{lemma: f2} 
If $\varphi$ is defined by \eqref{varphi2} with $\lambda=0$. 
Then its convex envelope is given by 
\begin{equation*}
\delta ( s ; 0)=
\begin{cases}
\varphi ( s-; 0), ~ & 0 \le s \le t_0,\\
\varphi ( t_0-; 0)+\varphi'( t_0 ; 0) (s-t_0 ), ~ & t_0 < s \le 1,
\end{cases}
\end{equation*}
where $t_0$ is the unique number in $(0, w(\alpha))$ such that 
	\begin{equation*}
	\varphi'( t_0 ; 0)=	\frac{ 1-\varphi ( t_0-; 0) }{ 1-t_0}.
	\end{equation*}
\end{lemma}
\begin{proof}
When $\lambda=0$, we have
\begin{equation*}
\varphi ( s-; 0)=
\left \{
\begin{aligned}
&\frac{1}{\alpha}w^{-1} (s), ~ & 0 \le s \le w(\alpha),\\
&1 , ~ & w(\alpha) < s \le 1,
\end{aligned}
\right.
\end{equation*} 
	Let
	$$f_2 (s)=1-\varphi (s-; 0)-\varphi'( s ; 0) ( 1-s ) , ~ 0<s<w(\alpha).$$
	For $0<s_1<s_2<w(\alpha)$, by integration by parts	we have
	\begin{equation*}
	\begin{aligned}
	f_2 (s_1)-f_2 (s_2) 
	=&-\varphi (s_1-; 0)+\varphi (s_2-; 0)-\varphi'( s_1 ; 0) ( 1-s_1 )+\varphi'( s_2 ; 0) ( 1-s_2 ) \\
	=& \int_{s_1}^{s_2} \varphi'( z ; 0) dz-\int_{s_1}^{s_2} d\big(z\varphi'( z ; 0)\big) +\int_{s_1}^{s_2} d \varphi'( z ; 0)\\
	=& \int_{s_1}^{s_2} ( 1-z ) d\varphi'( z ; 0) >0,
	\end{aligned}
	\end{equation*}
	again thanks to the convexity of $\varphi( z ; 0)$ on $[0, w(\alpha)]$, 
	so $f_2 (s)$ is a strictly decreasing function of $s$. Moreover, thanks to the strictly convexity of $w^{-1} (s)$, 
	\begin{equation*}
	\lim_{s \uparrow w(\alpha)} f_2 (s)=-\frac{1-w(\alpha)}{\alpha}\lim_{s \uparrow w(\alpha)}(w^{-1} (s))'< 0,
	\end{equation*}
	and
	\begin{equation*}
	\lim_{s \downarrow 0} f_2 (s)=1-\frac{1}{\alpha}\lim_{s \downarrow 0}(w^{-1} (s))'=1.
	\end{equation*}
	 Therefore, there exists a unique $t_0\in (0, w(\alpha))$ such that $f_2 ( t_0)=0$. By the same argument as in the proof of Lemma \ref{lemma: f0}, we can show that $\delta ( s-; 0)$ is the convex envelope of $\varphi ( s-; 0)$.
The left panel of Figure \eqref{figure:ES varphi delta} gives a graphical illustration.
\end{proof}

\begin{figure}[H]
\centering
\begin{minipage}[t]{3in}
\centering
\includegraphics[width=3in]{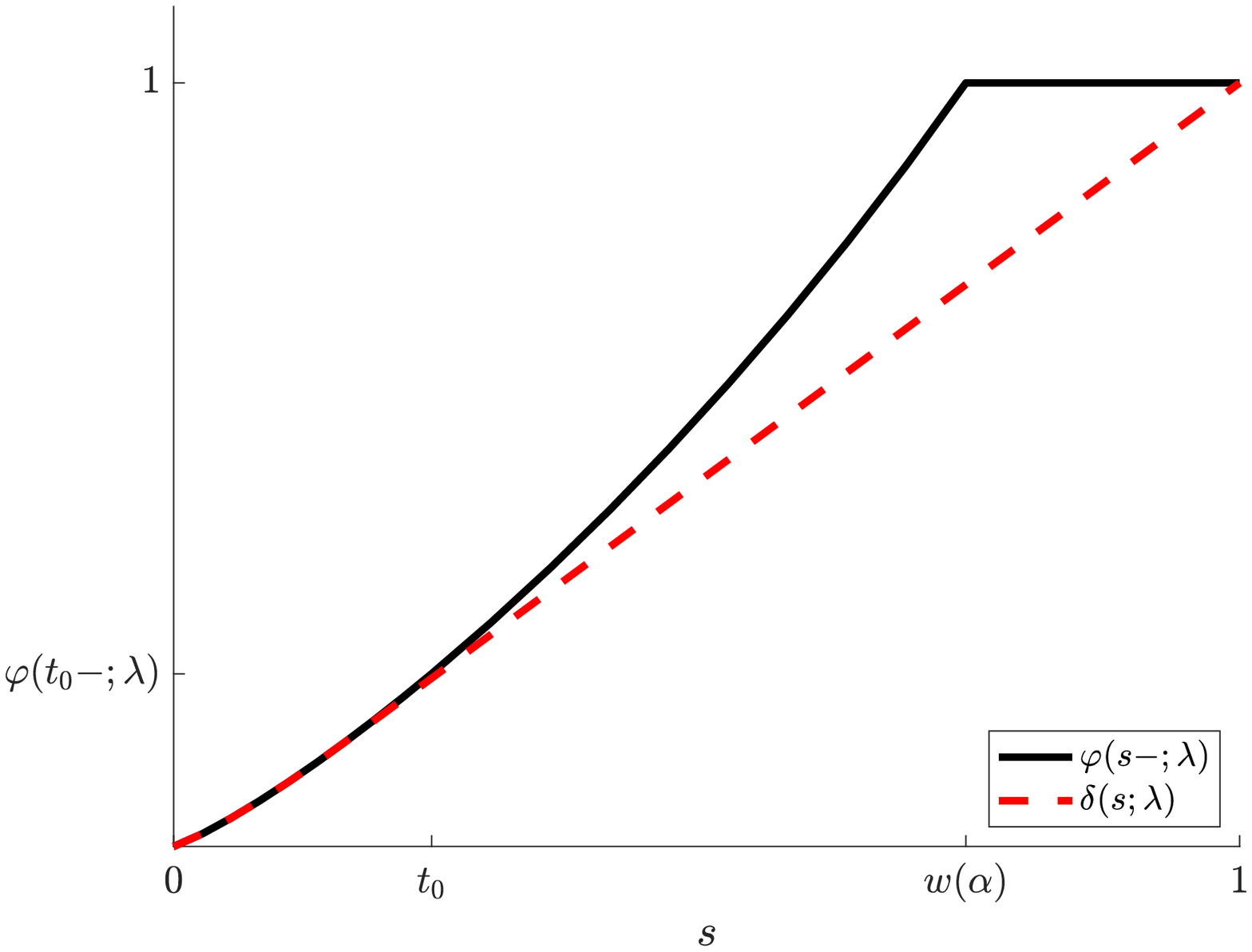}
\centerline{$\lambda=0$}
\end{minipage}
\begin{minipage}[t]{3in}
\centering
\includegraphics[width=3in]{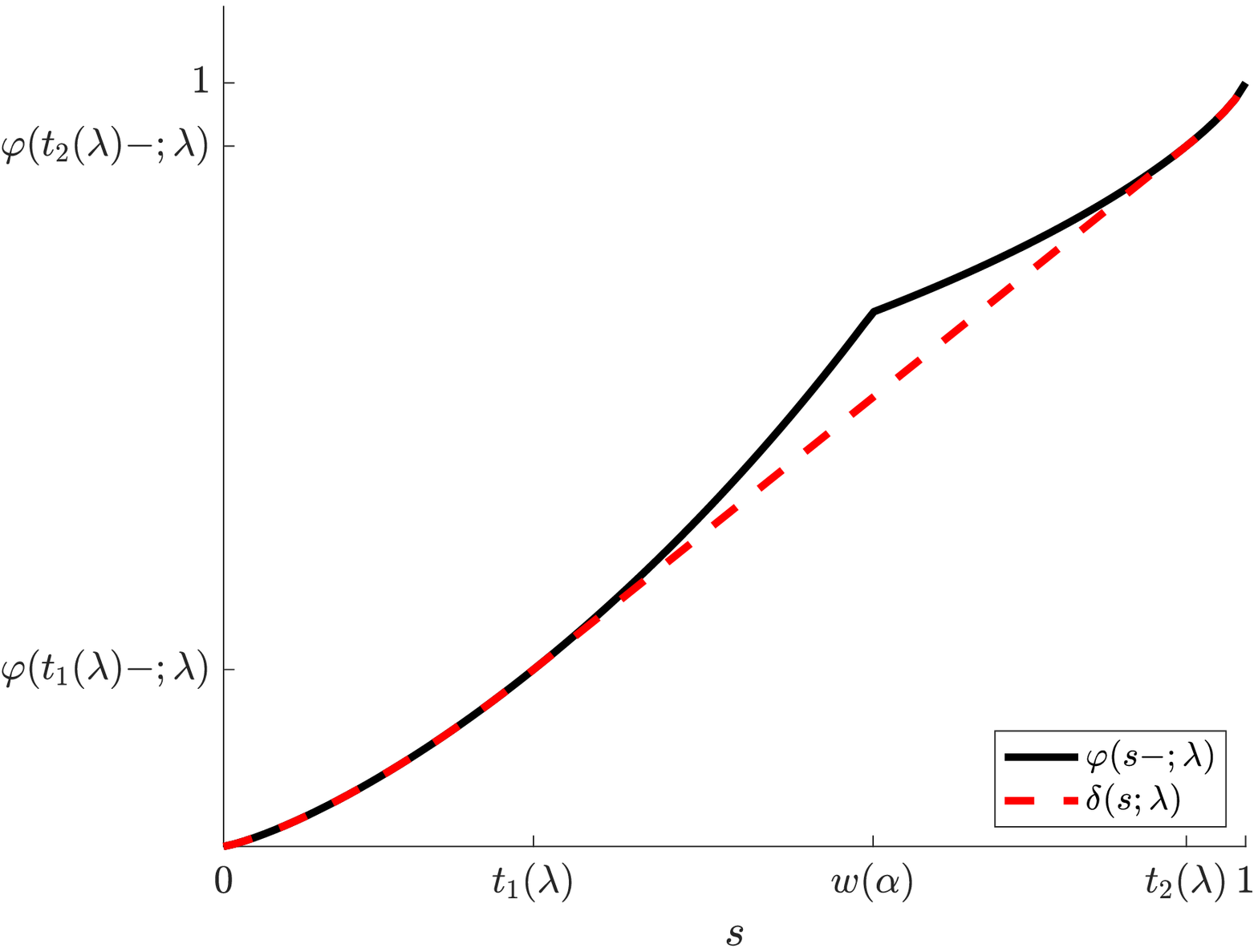}
\centerline{$0<\lambda<\infty$}\bigskip 
\end{minipage} 
\caption{$\varphi (s-,\lambda)$ and $\delta (s,\lambda)$ for ES }\label{figure:ES varphi delta}
\end{figure}

\begin{lemma}\label{lemma: f3} 
If $\varphi$ is defined by \eqref{varphi2} with $0<\lambda<\infty$. 
Then its convex envelope is given by 
\begin{equation*}
\delta ( s ; \lambda)=
\begin{cases}
\varphi ( s-; \lambda), ~ & 0 \le s \le t_1(\lambda),\\
\varphi ( t_1(\lambda)-; \lambda)+\varphi'( t_1(\lambda) ; \lambda) (s-t_1(\lambda) ), ~ & t_1(\lambda) < s \le t_2(\lambda),\\
\varphi ( s-; \lambda) , ~ & t_2(\lambda) < s \le 1,
\end{cases}
\end{equation*}
where $0<t_1 (\lambda)<w(\alpha)<t_2 (\lambda)<1$ are the unique pair such that 
	\begin{equation}\label{eq: t1 t2}
		\begin{cases}
		\varphi'( t_1(\lambda) ; \lambda) =\varphi'( t_2(\lambda) ; \lambda) \\
		\varphi'( t_1(\lambda) ; \lambda) =\frac{ \varphi ( t_2-; \lambda)-\varphi ( t_1-; \lambda)}{t_2-t_1}.
		\end{cases}
	\end{equation}
\end{lemma}
\begin{proof} 
Again by the same argument as in the proof of Lemma \ref{lemma: f0}, we can show that $\delta ( s-; \lambda)$ is the convex envelope of $\varphi ( s-; \lambda)$. The right panel of Figure \eqref{figure:ES varphi delta} gives a graphical illustration. It is only left to show the existence and uniqueness of the pair $0<t_1 (\lambda)<w(\alpha)<t_2 (\lambda)<1$ that satisfies \eqref {eq: t1 t2}.

Let 
$$\underline{s}(\lambda)=w\left( 1-F_{\xi} \left(\tfrac{\frac{1}{\alpha}+\lambda}{\lambda} G_{\xi} \left(1-\alpha \right) \right) \right).$$
In view of the strict monotonicity of $w$ and $F_{\xi}$ and positivity of $G_{\xi}$, we have $$\underline{s}(\lambda) < w\left( 1-F_{\xi} \big( G_{\xi} (1-\alpha ) \big) \right)=w(\alpha).$$

Define
	$$y(s)=w\left( 1-F_{\xi} \left(\tfrac{\lambda}{\frac{1}{\alpha}+\lambda} G_{\xi} \big(1-w^{-1}(s)\big) \right) \right), \quad s \in (\underline{s}(\lambda), w(\alpha)).$$
Then by monotonicity, for $s \in (\underline{s}(\lambda),w(\alpha))$,
\begin{align*}
y(s) &=w\left( 1-F_{\xi} \left(\tfrac{\lambda}{\frac{1}{\alpha}+\lambda} G_{\xi} \big(1-w^{-1}(s)\big) \right) \right) \\
&> w\left( 1-F_{\xi} \left(\tfrac{\lambda}{\frac{1}{\alpha}+\lambda} G_{\xi} \big(1-w^{-1}(\underline{s}(\lambda))\big) \right) \right) \\
&= w\left( 1-F_{\xi} \left(\tfrac{\lambda}{\frac{1}{\alpha}+\lambda} G_{\xi} \left( F_{\xi} \left(\tfrac{\frac{1}{\alpha}+\lambda}{\lambda} G_{\xi} \left(1-\alpha \right) \right) \right) \right) \right) \\
&= w\left( 1-F_{\xi} \left(G_{\xi} \left(1-\alpha \right) \right) \right) \\
& = w(\alpha).
\end{align*}
and it is not hard to verify that 
\begin{align}\label{eqprime}
\varphi'( s ; \lambda)=\varphi'( y(s) ; \lambda).
\end{align}	

	Let
	$$f_3 (s; \lambda)=\varphi (y(s)-; \lambda)-\varphi ( s-; \lambda)-\varphi'( s ; \lambda) ( y(s)-s ), ~ 0<s<w(\alpha).$$
	For $\underline{s}(\lambda) < s_1 < s_2 < w(\alpha)$, thanks to \eqref{eqprime} and using integration by parts, we have
	\begin{align*}
	&\quad\; f_3 (s_1; \lambda)-f_3 (s_2; \lambda) \\
	&=\varphi (y(s_1)-; \lambda)-\varphi (y(s_2)-; \lambda)-( \varphi ( s_1-; \lambda)-\varphi ( s_2-; \lambda) ) \\
	&\quad\;-\varphi'( s_1 ; \lambda) ( y(s_1)-s_1 )+\varphi'( s_2 ; \lambda) ( y(s_2)-s_2 ) \\
	&=-\int_{y(s_1)}^{y(s_2)} d\varphi( z-; \lambda)+\int_{s_1}^{s_2} d\big(\varphi'( z ; \lambda) y(z) \big)+\int_{s_1}^{s_2} d\varphi( z- ; \lambda) -\int_{s_1}^{s_2} d\big(\varphi'( z ; \lambda)z\big)\\
	&=
	-\int_{y(s_1)}^{y(s_2)} \varphi'( z ; \lambda) dz
	+\int_{s_1}^{s_2} d\big(\varphi'( y(z) ; \lambda) y(z) \big)-\int_{s_1}^{s_2}zd \varphi'( z ; \lambda) \\
	&=-\int_{y(s_1)}^{y(s_2)} \varphi'( z ; \lambda) dz 
	+\int_{y(s_1)}^{y(s_2)} d\big(\varphi'( z ; \lambda) z\big)-\int_{s_1}^{s_2}zd \varphi'( z ; \lambda) \\
	&=\int_{y(s_1)}^{y(s_2)} zd\varphi'( z ; \lambda)-\int_{s_1}^{s_2}zd \varphi'( z ; \lambda) \\ 
	&=\int_{s_1}^{s_2} y(z)d\varphi'( y(z) ; \lambda)+\int_{s_1}^{s_2}zd \varphi'( z ; \lambda)\\
	&=\int_{s_1}^{s_2}(y(z)-z)d \varphi'( z ; \lambda)>0
	\end{align*}
	by virtue of $y(z)>w(\alpha)>z$ for $z\in (\underline{s}(\lambda), w(\alpha))$ and $\varphi$ is convex on $(0, w(\alpha))$, 
	so $f_3 (s; \lambda)$ is a strictly decreasing function of $s$. Moreover, thanks to the strictly convexity of $\varphi$ $(0, w(\alpha))$ and on $(w(\alpha), 1)$, respectively, and the fact that $\lim_{s \downarrow \underline{s}(\lambda)} y(s) = w(\alpha)$, we have
	\begin{equation*}
\begin{aligned}
\lim_{s \downarrow \underline{s}(\lambda)} f_3 (s; \lambda) =& \lim_{s \downarrow \underline{s}(\lambda)} \left( \int_{w(\alpha)}^{y(s)} \varphi'( z ; \lambda)dz + \int_s^{w(\alpha)} \varphi'( z ; \lambda)dz - \varphi'( s ; \lambda) \left( y(s) - s \right) \right) \\
=& \int_{\underline{s}(\lambda)}^{w(\alpha)} \varphi'( z ; \lambda)dz - \varphi'( \underline{s}(\lambda) ; \lambda) \left( w(\alpha) - \underline{s}(\lambda) \right) \\
=& \int_{\underline{s}(\lambda)}^{w(\alpha)} \left( \varphi'( z ; \lambda)- \varphi'( \underline{s}(\lambda) ; \lambda) \right)dz >0, 
\end{aligned}
	\end{equation*}
	and
\begin{equation*}
\begin{aligned}
\lim_{s \uparrow w(\alpha)} f_3 (s; \lambda) =& \lim_{s \uparrow w(\alpha)} \left( \int_{w(\alpha)}^{y(s)} \varphi'( z ; \lambda)dz + \int_s^{w(\alpha)} \varphi'( z ; \lambda)dz - \varphi'( s ; \lambda) \left( y(s) - s \right) \right) \\
=& \int_{w(\alpha)}^{y(w(\alpha)-)} \varphi'( z ; \lambda)dz - \varphi'( w(\alpha)- ; \lambda) \left( y(w(\alpha)-) - w(\alpha) \right) \\
=& \int_{w(\alpha)} ^{y(w(\alpha)-)} \left( \varphi'( z ; \lambda)-\varphi'( y(w(\alpha)-) ; \lambda) \right)dz < 0.
\end{aligned}
\end{equation*}
	Therefore, there exists a unique $t_1 (\lambda) \in (\underline{s}(\lambda),w(\alpha))$ that solves $f_3 ( t_1 (\lambda) ; \lambda)=0$. Let $t_2 (\lambda)=y(t_1 (\lambda))$. Then, $t_1 (\lambda)$ and $t_2 (\lambda)$ solve \eqref{eq: t1 t2}.
\end{proof}

\section{Proofs}\label{sec:A2}

\begin{proof}[Proof of Lemma \ref{lemma:3.1}]
This is indeed contained in the proof of Theorem 7, \cite{xu2014}. We give it here to make the paper self-contained. 
Suppose $X_T^{*}$ is optimal to \eqref{prob:martingale} and $ \E [\xi_T X_T^{*}] < x$. Define $$\tilde{X}_T=X_T^{*}+\frac{x-\E [\xi_T X_T^{*}]}{\E [\xi_T]},$$ 
which satisfies the budget constraint. Clearly $\tilde{X}_T>X_T^{*}$. 
Let $R_T^{*}$ and $\tilde{R}_T$ be the corresponding log-returns of $X_T^{*}$ and $\tilde{X}_T$, respectively. Then 
\begin{equation*}
\begin{aligned}
\tilde{R}_T &> R_T^{*}, ~ a.s.,\\
\E [ \tilde{R}_T] &> \E [ R_T^{*}],\\
\rho _{\Phi} (\tilde{R}_T) &< \rho _{\Phi} (R_T^{*}), 
\end{aligned}
\end{equation*}
and consequently
$$ \lambda \E [ \tilde{R}_T]-\rho _{\Phi} (\tilde{R}_T) > \lambda \E [R_T^{*}]-\rho _{\Phi} (R_T^{*}), $$ contradicting the optimality of $X_T^{*}$ to \eqref{prob:martingale}. 
\end{proof}


\begin{proof}[Proof of Proposition \ref{prop:3.1}]
We handle the two cases $0<\lambda<\infty$ and $\lambda=0$ separately.	
\begin{description}
\item[Case $0<\lambda<\infty$.] Define
\begin{equation*}
\mathbb{H}_1=\big\{H \in \mathbb{G}\colon H(s)>0, ~ \forall s\in (0,1] \big\} \cap \mathbb{H}_{\Phi}.
\end{equation*}
It suffices to consider the following optimization problem
\begin{equation*}
\max_{H \in \mathbb{H}_1 } L(H(\cdot) ; \lambda, \eta),
\end{equation*}
as $L(H(\cdot) ; \lambda, \eta)=-\infty,$ if $H \in 	\mathbb{G} \setminus	\mathbb{H}_1$.

Because $\delta (s; \lambda)$ is the convex envelope of $\varphi ( s-; \lambda)$, $\varphi ( s-; \lambda) \ge \delta ( s ; \lambda), ~ s \in[0,1]$, $\delta (0; \lambda)=\varphi ( 0-; \lambda)=0$, and $\delta (1; \lambda)=\varphi ( 1-; \lambda)=1$. For any $H \in \mathbb{H}_1$, we have
\begin{equation*}
\int _{[0,1)} ( \varphi (s-; \lambda)-\delta (s; \lambda) ) d \ln H(s) \ge 0.
\end{equation*}
From Fubini's theorem, we have
\begin{equation*}
\begin{aligned}
&\int _{[0,1)} \big( \varphi (s-; \lambda)-\delta (s; \lambda) \big) d \ln H(s) \\
=& \int _{[0,1)} \ln H(s) d \delta (s; \lambda)-\int _{[0,1)} \ln H(s) d \varphi (s; \lambda) \ge 0.
\end{aligned}
\end{equation*}
and consequently
\begin{align*}
L(H(\cdot) ; \lambda, \eta)&= (1+\lambda) \int _{[0,1)} \ln H(s) d\varphi ( s ; \lambda)-\eta \int_{[0,1)} H (s) ds \\
&\le (1+\lambda) \int _{[0,1)} \ln H(s) d \delta (s; \lambda)-\eta \int_{[0,1)} H (s) ds \\
&= (1+\lambda) \int _{[0,1)} \ln \left( H(s) \right) \delta' (s; \lambda) d s-\eta \int_{[0,1)} H (s) ds \\ 
&\le (1+\lambda) \int _{[0,1)} \ln \left( H^{*} (s; \lambda , \eta) \right) \delta' (s; \lambda) d s-\eta \int_{[0,1)}H^{*} (s; \lambda , \eta) ds,
\end{align*}
where $H^{*} (s; \lambda , \eta)=\frac{1+\lambda}{\eta } \delta' (s; \lambda)$ is given by point-wise optimization. We then show that 
$$L( H^{*} ( \cdot ; \lambda , \eta) ; \lambda, \eta)=(1+\lambda) \int _{[0,1)} \ln \left( H^{*} (s; \lambda , \eta) \right) \delta' (s; \lambda) d s-\eta \int_{[0,1)}H^{*} (s; \lambda , \eta) ds,$$
which is equivalent to
\begin{equation*} 
\int _{[0,1)} \ln H^{*} (s; \lambda , \eta) d\varphi ( s ; \lambda)=\int _{[0,1)} \ln H^{*} (s; \lambda , \eta) d \delta (s; \lambda), 
\end{equation*}
or by Fubini's theorem, 
\begin{equation*}
\int _{(0,1)} \big( \varphi (s-; \lambda)-\delta (s; \lambda) \big) d \ln H^{*} (s; \lambda , \eta)=0,
\end{equation*}
namely,
\begin{equation}\label{eq:verify equality}
\int _{(0,1)} \big( \varphi (s-; \lambda)-\delta (s; \lambda) \big) d \ln \big( \delta' (s; \lambda) \big)=0.
\end{equation}

Because $\delta' (s; \lambda)$ is constant on any sub-interval of $\{s \in (0, 1) : \varphi (s-; \lambda) > \delta (s; \lambda) \}$, the above identity \eqref{eq:verify equality} holds.

Finally, we need to verify $H^{*} (\cdot; \lambda , \eta) \in 	\mathbb{H}_1$. If $\Phi (\{0\})=0$, then
\begin{equation*}
\mathbb{H}_1=\big \{H \in \mathbb{G} \colon H(s)>0, ~ \forall s\in (0,1] \big\};
\end{equation*}
if $\Phi (\{0\})>0$, then
\begin{equation*}
\mathbb{H}_1=\big \{H \in \mathbb{G} \colon H(0)>0\big\}.
\end{equation*}

Because $0<\lambda<\infty$, $\varphi (s-; \lambda)>0$ for $s>0$. We claim $ \delta' (s; \lambda)>0$ for $s>0$. Otherwise, we have 
$$s_1=\sup \big\{ s \in [0,1]\colon \delta' (s; \lambda)=0 \big\}>0.$$
Since 
$\delta (1; \lambda) >\delta (0; \lambda) $, we have $s_1 <1$. It follows from the convexity of $\delta (s; \lambda)$ that $ \delta' (s; \lambda)=0$ for $s<s_1$ and $ \delta' (s; \lambda)>0$ for $s>s_1$. This means $\delta (s; \lambda)$ is not affine around $s_{1}$. 
But $\varphi (s_{1}-; \lambda)>0=\delta (s_{1}; \lambda)$, so $\delta (s; \lambda)$ should be affine in the neighborhood of $s_1$, leading to a contradiction.

We need to additionally show $\delta' (0; \lambda)>0$ if $\Phi (\{0\})=0$. In fact, $\varphi (0; \lambda)> \varphi (0-; \lambda)=\delta (0; \lambda)=0$. If $\delta' (0; \lambda)=0$, then $\varphi (s-; \lambda)> \delta (s; \lambda)$ and $\delta (s; \lambda)$ is affine for $s$ sufficiently close to $0$, which contradicts the fact $ \delta' (s; \lambda)>0$ for $s>0$.

\item[Case $\lambda=0$.] In view of \eqref{eq:-infty integral}, we need to solve the following optimization problem
\begin{equation*}
\max_{H \in \mathbb{H}_{\Phi} } ~L(H(\cdot) ; 0 , \eta).
\end{equation*}
Let
\begin{equation*}
z_{\Phi}=\sup \big\{ z \in [0,1]\colon \Phi ( [0,z] )=0 \big\}.
\end{equation*}
We then have $\Phi ([0,z_{\Phi} ))=0$ and $\varphi ( z-; 0)=\Phi ([0,w^{-1} (z)))=0$, for all $z \in [0, w(z_{\Phi} )]$. According to \eqref{eq:concave envelope}, we know that $\delta ( z; 0)=\varphi ( z-; 0)=0$, for all $z \in [0, w(z_{\Phi} )]$. There are two subcases.
\begin{itemize}
\item If $\Phi ( \{ z_{\Phi} \}) >0$, then $\mathbb{H}_{\Phi}=\{H\in \mathbb{G}\colon H(w(z_{\Phi}))>0 \}$
and
\begin{equation*}
\begin{aligned}
L(H(\cdot) ;0, \eta)=& \int _{[w(z_{\Phi}),1)} \ln H(s) d\varphi ( s ; 0)-\eta \int_{[0,1)} H (s) ds \\
=& \int _{[w(z_{\Phi}),1)} \ln H(s) d\Phi ([0,w^{-1} (s)])-\eta \int_{[0,1)} H (s) ds. 
\end{aligned}
\end{equation*}
Similar to the case $0<\lambda<\infty$, we can show, for any $H \in \mathbb{H}_{\Phi}$,
\begin{equation*}
\begin{aligned}
& \int _{[w(z_{\Phi}),1)} \ln H(s) d \delta (s; \lambda)-\int _{[w(z_{\Phi}),1)} \ln H(s) d \varphi (s; \lambda) \\
=&\int _{[w(z_{\Phi}),1)} \big( \varphi (s-; \lambda)-\delta (s; \lambda) \big) d \ln H(s) \ge 0.
\end{aligned}
\end{equation*}
We then have 	
\begin{equation*}
\begin{aligned}
L(H(\cdot) ; 0, \eta)=& \int _{[w(z_{\Phi}),1)} \ln H(s) d\varphi ( s ; 0)-\eta \int_{[0,1)} H (s) ds \\
\le & \int _{[w(z_{\Phi}),1)} \ln H(s) d\delta ( s ; 0)-\eta \int_{[0,1)} H (s) ds \\
=& \int _{[w(z_{\Phi}),1)} \ln \left( H(s) \right) \delta' (s; 0) d s-\eta \int_{[0,1)} H (s) ds \\ 
\le & \int _{[w(z_{\Phi}),1)} \ln \left( H^{*} (s; 0 , \eta) \right) \delta' (s; \lambda) d s-\eta \int_{[0,1)}H^{*} (s; 0 , \eta) ds,
\end{aligned}
\end{equation*}
where 
\begin{equation*}
H^{*} (s; \lambda , \eta)=\frac{1}{\eta } \delta' (s; \lambda)=
\begin{cases}
0, ~ &0 \le z < w(z_{\Phi}), \\
\frac{1}{\eta } \delta' (s; \lambda), ~ & w(z_{\Phi}) \le z \le 1,
\end{cases}
\end{equation*}
is given by point-wise optimization. We can similarly show
\begin{equation*}
L( H^{*} ( \cdot ; 0 , \eta) ; 0, \eta)=\int _{[w(z_{\Phi}),1)} \ln \left( H^{*} (s; 0 , \eta) \right) \delta' (s; 0) d s-\eta \int_{[0,1)}H^{*} (s; 0 , \eta) ds
\end{equation*}
by verifying
\begin{equation*}
\int _{(w(z_{\Phi}),1)} \big( \varphi (s-; 0)-\delta (s; 0) \big) d \ln \big( \delta' (s; 0) \big)=0.
\end{equation*}

We now verify $H^{*} (\cdot; \lambda , \eta) \in 	\mathbb{H}_{\Phi}$ by showing $\delta' (w(z_{\Phi}); \lambda)>0$. Similar to the case $0<\lambda<\infty$, we can show $\delta' (s; 0)>0$ for $ s>w(z_{\Phi})$. Because $\varphi (w(z_{\Phi}); 0)> \varphi (w(z_{\Phi})-; 0)=\delta (w(z_{\Phi});0)=0$, we can prove $\delta' (w(z_{\Phi}); 0)>0$ in a similar fashion.

\item If $\Phi ( \{ z_{\Phi} \})=0$, then $\delta ( w(z_{\Phi} ); 0)=\varphi ( w(z_{\Phi} )-; 0)=\varphi ( w(z_{\Phi} ) ; 0)=0$, $\mathbb{H}_{\Phi}=\{H\in \mathbb{G}\colon H(s)>0, ~ \forall s>w(z_{\Phi}) \}$
and
\begin{equation*}
L(H(\cdot) ;0, \eta)=\int _{(w(z_{\Phi}),1)} \ln H(s) d\varphi ( s ; 0)-\eta \int_{[0,1)} H (s) ds.
\end{equation*}
Similar to the case $\lambda=0$ and $\Phi ( \{ z_{\Phi} \}) >0$, we can show, for any $H \in \mathbb{H}_{\Phi}$,
\begin{align*}
L(H(\cdot) ; 0, \eta)&= \int _{(w(z_{\Phi}),1)} \ln H(s) d\varphi ( s ; 0)-\eta \int_{[0,1)} H (s) ds \\
&\le \int _{(w(z_{\Phi}),1)} \ln H(s) d\delta ( s ; 0)-\eta \int_{[0,1)} H (s) ds \\
&\le \int _{(w(z_{\Phi}),1)} \ln \left( H^{*} (s; 0 , \eta) \right) \delta' (s; \lambda) d s-\eta \int_{[0,1)}H^{*} (s; 0 , \eta) ds,
\end{align*}
where 
\begin{equation*}
H^{*} (s; \lambda , \eta)=\frac{1}{\eta } \delta' (s; \lambda)=
\begin{cases}
0, ~ &0 \le z < w(z_{\Phi}), \\
\frac{1}{\eta } \delta' (s; \lambda), ~ & w(z_{\Phi}) \le z \le 1,
\end{cases}
\end{equation*}
is given by point-wise optimization. We can similarly show
$$L( H^{*} ( \cdot ; 0 , \eta) ; 0, \eta)=\int _{(w(z_{\Phi}),1)} \ln \left( H^{*} (s; 0 , \eta) \right) \delta' (s; 0) d s-\eta \int_{[0,1)}H^{*} (s; 0 , \eta) ds.$$
Finally, it is easy to verify $\delta' (s; \lambda)>0$ for $s>w(z_{\Phi})$ and thus $H^{*} (\cdot; 0 , \eta) \in \mathbb{H}_{\Phi}$.
\end{itemize}
\end{description} 
The proof is complete. 
\end{proof}

\begin{proof}[Proof of Proposition \ref{prop:3.2}]
For any $H \in \mathbb{G}$ that satisfies the budget constraint in \eqref{prob:quantile H}, we have
\begin{align*}
& \quad\; \lambda \int_{[0,1)} \ln H(s) dw^{-1} (s)+\int _{[0,1)} \ln H(s) d \Phi ([0,w^{-1} (s)]) \\
&=L\left(H(\cdot) ; \lambda , \frac{1+\lambda}{x}\E [\xi_T] \right)+\frac{1+\lambda}{x}\E [\xi_T] \int_{[0,1)} H(s)ds \\
&=L\left(H(\cdot) ; \lambda , \frac{1+\lambda}{x}\E [\xi_T] \right)+(1+\lambda). 
\end{align*}
Similarly, 
\begin{align*}
&\quad\; \lambda \int_{[0,1)} \ln H^{*} \left(s; \lambda , \frac{1+\lambda}{x}\E [\xi_T] \right) dw^{-1} (s) \\
&\quad\quad+\int _{[0,1)} \ln H^{*} \left(s; \lambda , \frac{1+\lambda}{x}\E [\xi_T] \right) d \Phi ([0,w^{-1} (s)])\\
&=L\left(H^{*} \left(\cdot; \lambda , \frac{1+\lambda}{x}\E [\xi_T] \right) ; \lambda , \frac{1+\lambda}{x}\E [\xi_T]\right)+(1+\lambda). 
\end{align*}
By virtue of Proposition \ref{prop:3.1}, the claim is proved. 
\end{proof}


\begin{proof}[Proof of Proposition \ref{prop:4.1}] 
First, we have
\begin{equation*}
\varphi ( s-; \lambda)=\frac{ \Phi ([0,w^{-1} (s)))+\lambda w^{-1} (s)}{1+\lambda}=
\begin{cases}
\frac{ \lambda w^{-1} (s)}{1+\lambda}, ~ & 0 \le s \le w(\alpha),\\
\frac{ 1+\lambda w^{-1} (s)}{1+\lambda}, ~ & w(\alpha) < s \le 1.
\end{cases}
\end{equation*}

Next, let $\delta (s; \lambda)$ be the convex envelope of $\varphi ( s-; \lambda)$, and $ \delta' (s; \lambda)$, the right derivative of $\delta (s; \lambda)$ with respect to $s$.
\begin{description}
\item[Case $\lambda=0$.] By Lemma \ref{lemma: f0}, we have 
\begin{equation*}
\delta'( s ; 0)=
\begin{cases}
0, ~ & 0 \le s < w(\alpha),\\
\frac{ 1}{ 1-w(\alpha)} , ~ & w(\alpha) \le s < 1.
\end{cases}
\end{equation*}
It follows from Proposition \ref{prop:efficient} that the optimal solution is given by
\begin{align*}
X^{\text{VaR}}_{T,0} &= \frac{x}{\E [\xi_T]} \delta' (w (1-F_{\xi}(\xi_T)); 0)\\
&=
\begin{cases}
0, ~ & \xi_T > G_{\xi} (1-\alpha),\\
\frac{x}{\E [\xi_T]} \cdot \frac{ 1}{ 1-w(\alpha)} , ~ & \xi_T \le G_{\xi} (1-\alpha).
\end{cases}	
\end{align*}

\item[Case $0<\lambda<\infty$.] By Lemma \ref{lemma: f1}, we have 
\begin{equation*}
\delta' ( s; \lambda)=
\begin{cases}
\varphi'( s ; \lambda), ~ & 0 \le s < w(\alpha),\\
\varphi'( s^{*}(\lambda) ; \lambda), ~ & w(\alpha) < s \le s^{*}(\lambda),\\
\varphi'( s ; \lambda) , ~ & s^{*}(\lambda) < s < 1.
\end{cases}
\end{equation*}
It follows that 
\begin{align*}
X^{\text{VaR}}_{T,\lambda}&= \frac{x}{\E [\xi_T]} \delta' (w (1-F_{\xi}(\xi_T)); \lambda)\\
&=
\begin{cases}
\frac{\lambda}{1+\lambda} \cdot \frac{x}{\xi_T} , ~ & \xi_T > G_{\xi} (1-\alpha),\\
\frac{\lambda}{1+\lambda} \cdot \frac{x}{ G_{\xi} (1-w^{-1}(s^{*}(\lambda) ))} , ~ & G_{\xi} (1-w^{-1}(s^{*}(\lambda) )) < \xi_T \le G_{\xi} (1-\alpha) ,\\
\frac{\lambda}{1+\lambda} \cdot \frac{x}{\xi_T} , ~ & \xi_T \le G_{\xi} (1-w^{-1}(s^{*}(\lambda) )).
\end{cases}
\end{align*}
\end{description}

The rest of the claim is straightforward to verify.
\end{proof}


\begin{proof}[Proof of Corollary \ref{coro:4.1}]
Conditional on $\mathcal{F}_t$, $\ln \xi_T $ is normally distributed with mean $\ln \xi_t- \left(r+\frac{\theta ^2 }{2} \right)(T-t) $ and variance $\theta ^2 (T-t)$.

For $\lambda=0$, we have 
\begin{equation*}
X^{\text{VaR}}_{t,0}=\frac{1}{\xi_t} \E \left[ \xi_T X^{\text{VaR}}_{T,0}\;|\; \mathcal{F}_t\right]=\frac{1}{\xi_t} \E \left[ \xi_T \underline{X}_{\text{VaR}} \mathbf{1}_{ \xi_T \le \xi_{\alpha} }\;|\; \mathcal{F}_t\right]=\underline{X}_{\text{VaR}} e^{-r (T-t)} N \left( d_2 (t, \xi_t, \xi_{\alpha}) \right).
\end{equation*}
Applying Ito's lemma to $X^{\text{VaR}}_{t,0}$, we have
\begin{equation*}
\begin{aligned}
d X^{\text{VaR}}_{t,0}=(\cdots)dt+\frac{\underline{X}_{\text{VaR}} e^{-r (T-t)} \nu \left( d_2 (t, \xi_t, \xi_{\alpha}) \right)}{ \sqrt{T-t} } dW_t. 
\end{aligned}
\end{equation*}
Comparing the coefficient of the $dW_t$ term with \eqref{eq:budget}, we arrive at the expression for $\pi^{\text{VaR}}_{t,0}$.

For $0<\lambda<\infty$, we have 
\begin{align*}
X^{\text{VaR}}_{t,\lambda}&=\frac{1}{\xi_t} \E\left[\xi_T X^{\text{VaR}}_{T,\lambda}\;\big|\; \mathcal{F}_t\right] \\ 
&=\frac{1}{\xi_t} \E \left[ \frac{\lambda}{1+\lambda} x \mathbf{1}_{\xi_T > \xi_{\alpha}}+\xi_T \underline{X}_{\text{VaR}} \mathbf{1}_{\underline{\xi}_{\text{VaR}} < \xi_T \le \xi_{\alpha}}+\frac{\lambda}{1+\lambda} x \mathbf{1}_{\xi_T \le \underline{\xi}_{\text{VaR}} }\;\bigg|\; \mathcal{F}_t\right]\\
&=\frac{1}{\xi_t} \cdot \frac{\lambda}{1+\lambda} x N \left(-d_1 (t, \xi_t, \xi_{\alpha}) \right)\\
&\quad\;+\underline{X}_{\text{VaR}} e^{-r(T-t)} \left( N \left( d_2 (t, \xi_t, \xi_{\alpha}) \right)-N \left( d_2 ( t, \xi_t, \underline{\xi}_{\text{VaR}} ) \right) \right)\\
&\quad\;+\frac{1}{\xi_t} \cdot \frac{\lambda}{1+\lambda} x N \left( d_1 (t, \xi_t, \underline{\xi}_{\text{VaR}} ) \right).
\end{align*}
Applying Ito's lemma to $X^{\text{VaR}}_{t,\lambda}$ and comparing the coefficient of the $dW_t$ term with \eqref{eq:budget}, we have
\begin{align*}
\pi^{\text{VaR}}_{t,\lambda}&= \frac{\lambda}{1+\lambda} \cdot \frac{x}{\xi_t} \cdot \Bigg( N\left(-d_1(t, \xi_t, \xi_{\alpha}) \right) \frac{\theta}{\sigma}-\frac{\nu \left( d_1 (t, \xi_t, \xi_{\alpha}) \right) }{\sigma \sqrt{T-t}}\\
&\qquad\qquad\qquad\quad+N\big(-d_1(t, \xi_t, \underline{\xi}_{\text{VaR}}) \big) \frac{\theta}{\sigma}+\frac{\nu \big( d_1 (t, \xi_t, \underline{\xi}_{\text{VaR}} ) \big)}{\sigma \sqrt{T-t}} \Bigg) \\
&\quad\;+\frac{\underline{X}_{\text{VaR}} e^{-r (T-t)} }{\sigma \sqrt{T-t} } \cdot \left( \nu \left( d_2 ( t, \xi_t, \xi_{\alpha}) \right)-\nu \big( d_2 (t, \xi_t, \underline{\xi}_{\text{VaR}}) \big) \right).
\end{align*}
Noting that
$$ \nu (d_1 (t, \xi_t, y) )=e^{-r (T-t)} \frac{\xi_t}{y}\nu (d_2 (t, \xi_t, y) ),$$
we arrive at the expression for $\pi^{\text{VaR}}_{t,\lambda}$.
\end{proof}


\begin{proof}[Proof of Proposition \ref{prop:4.2}]
As the proof is similar to that of Proposition \ref{prop:4.1} by appealing to Lemmas \ref{lemma: f2} and \ref{lemma: f3}, we omit the details. In view of the fact that $\ln \xi_T $ is normally distributed with mean $-\left( r+\frac{\theta ^2 }{2} \right)T $ and variance $\theta ^2 T$, it is straightforward to obtain $\text{ES}_{\alpha}(R^{\text{ES}}_{T,0})$ and $\text{ES}_{\alpha}(R^{\text{ES}}_{T,\lambda})$.
\end{proof}

\bibliographystyle{apalike}

\end{document}